%% file: main.tex
\newif\ifsoda
\sodafalse
\ifsoda
	\documentclass[twoside,leqno,twocolumn]{article}
	\usepackage{ltexpprt}
\else
	\documentclass[11pt]{article}
	\usepackage{geometry}
\fi
\usepackage[utf8]{inputenc}
\usepackage{enumerate}
\usepackage{amssymb}
\usepackage{amsmath}
\usepackage[breaklinks,bookmarks=false]{hyperref}
\usepackage[usenames,dvipsnames,svgnames,table]{xcolor}
\usepackage{bbm}
\usepackage{balance}

\ifsoda
\else
	\usepackage[amsmath,amsthm,thmmarks,hyperref]{ntheorem}
\fi

\usepackage{nicefrac}
\usepackage{framed}
\usepackage[boxed,noline,longend,algo2e]{algorithm2e}

\usepackage{caption}
\usepackage{xspace}

\usepackage{tikz}
\usetikzlibrary{arrows,shapes}
\usetikzlibrary{positioning,calc}
\usetikzlibrary{decorations.pathreplacing}
\usetikzlibrary{fadings}
\usetikzlibrary{patterns}
\usetikzlibrary{snakes}
\tikzset{>= stealth'}
\tikzstyle{vertex}=[circle, draw,fill=gray!20, inner sep=0pt, minimum size=16pt]
\tikzstyle{square}=[rectangle, draw,fill=gray!20, inner sep=0pt, minimum size=16pt]
\tikzstyle{svertex}=[circle, draw,fill=gray!20, inner sep=0pt, minimum size=12pt]
\tikzstyle{ssquare}=[rectangle, draw,fill=gray!20, inner sep=0pt, minimum size=12pt]
\tikzstyle{tvertex}=[circle, draw,fill=gray!20, inner sep=0pt, minimum size=8pt]
\tikzstyle{tsquare}=[rectangle, draw,fill=gray!20, inner sep=0pt, minimum size=6pt]

\ifnum\pdfshellescape=1 \usetikzlibrary{external}  \fi

\ifsoda
	\newtheorem{definition}{Definition}
	\newtheorem{claim}{Claim} 

	\newtheorem{remark}{Remark}[section]

\else
	\newtheorem{theorem}{Theorem}[section]
	
	\newtheorem{lemma}[theorem]{Lemma}
	\newtheorem{fact}[theorem]{Fact}

	\newtheorem{definition}[theorem]{Definition}
	\newtheorem{claim}{Claim} 

	\qedsymbol{\openbox}
	
	\theoremstyle{empty}
	\newtheorem{artificialtheorem}{sdfsd}
\fi

\newcommand{\jobs}{\ensuremath\mathcal{J}\xspace}
\newcommand{\machines}{\ensuremath\mathcal{M}\xspace}

\newcommand{\eps}{\varepsilon}

\newcommand{\bR}{\mathbb{R}}

\newcommand{\rb}[1]{\left( #1 \right)} 
\renewcommand{\sb}[1]{\left[ #1 \right]} 

\newcommand{\ys}{y^\star}
\newcommand{\istar}{i^\star}

\newcommand{\Size}[1]{\ensuremath{\text{size}(#1)}}
\newcommand{\Cost}[1]{\ensuremath{\text{cost}(#1)}}

\DeclareMathOperator{\cost}{\textrm{cost}}
\DeclareMathOperator{\size}{\textrm{size}}
\newcommand{\yin}{y^{\mathrm{in}}}
\newcommand{\yout}{y^{\mathrm{out}}}

\begin{document}

\ifsoda\else
	\begin{titlepage}
\fi
\title{Unrelated Machine Scheduling of Jobs with Uniform Smith Ratios}

\author{
Christos Kalaitzis\thanks{School of Computer and Communication Sciences, EPFL. \newline
Email: \texttt{\{christos.kalaitzis,ola.svensson,jakub.tarnawski\}@epfl.ch}. \newline
Supported by ERC Starting Grant 335288-OptApprox.}
\and
Ola Svensson\footnotemark[1]
\and
Jakub Tarnawski\footnotemark[1]
}

\maketitle
\ifsoda\else
	\thispagestyle{empty}
\fi

\begin{abstract}
We consider the classic problem of scheduling jobs on unrelated machines so as to minimize the weighted sum of completion times.  Recently, for a small constant $\varepsilon >0 $, Bansal et al.  gave  a $(3/2-\varepsilon)$-approximation algorithm improving upon the ``natural'' barrier of $3/2$ which follows from independent randomized rounding. In simplified terms, their result is obtained by an enhancement of independent randomized rounding via strong negative correlation properties.

In this work, we take a different approach and  propose to use the same elegant rounding scheme  for the weighted completion time objective as devised by Shmoys and Tardos for optimizing a linear function subject to makespan constraints. Our main result is a $1.21$-approximation algorithm for the natural special case where the weight of a job is proportional to its processing time (specifically, all jobs have the same Smith ratio), which expresses the notion that each unit of work has the same weight. In addition, as a direct consequence of the rounding, our algorithm also achieves a bi-criteria $2$-approximation for the makespan objective. Our technical contribution is a tight analysis of the expected cost of the solution compared to the one given by the Configuration-LP relaxation -- we reduce this task to that of understanding certain worst-case instances which are simple to analyze.
\end{abstract}

\ifsoda\else
	\end{titlepage}
\fi
\section{Introduction} \label{sec:intro}
\input{intro.tex}

\input{prelim.tex}

\input{alg.tex}

\input{analysis.tex}

\input{hardness.tex}

\balance
\bibliographystyle{alphaurl}
\bibliography{references}


\end{document}

%% file: intro.tex
We study the classic problem of scheduling jobs on unrelated machines so as to
minimize the weighted sum of completion times.  Formally, we are given a set
$\machines$ of machines, and a set $\jobs$ of jobs, each with a weight
$w_j \geq 0$, such that the processing time (also called \emph{size}) of job $j$ on machine $i$ is $p_{ij}\geq
0$. The objective is to find a schedule which minimizes the weighted completion
time, that is $\sum_{j\in\jobs} w_j C_j$, where $C_j$ denotes the completion
time of job $j$ in the constructed schedule. In the three-field notation
used in scheduling literature~\cite{Graham79}, this problem is denoted as $R
| | \sum w_j C_j$.

The weighted completion time objective, along with makespan and flow time
minimization, is one of the most relevant and well-studied objectives for
measuring the quality of service in scheduling. Already in 1956, Smith~\cite{Smith56}
showed a simple rule for minimizing this objective on a single machine:
schedule the jobs in non-increasing order of  $w_j/p_j$  (where $p_j$ denotes
the processing time of job $j$ on the single machine). This order is often
referred to as the Smith ordering of the jobs and the ratio $w_j/p_j$ is called
the Smith ratio of job $j$. In the case of parallel machines, the problem
becomes significantly harder. Already for identical machines (the processing
time of a job is the same on all machines), it is strongly
NP-hard, and for the more general unrelated machine model that we consider,
the problem is NP-hard to approximate within $1+\varepsilon$, for a small
$\varepsilon>0$~\cite{HoogeveenSW01}.

Skutella and Woeginger~\cite{SkutellaW99} settled the approximability for identical
machines by developing a~polynomial time approximation scheme.  That is, for
every $\varepsilon >0$, they gave a $(1+\varepsilon)$-approximation algorithm
for minimizing the weighted sum of completion times on identical parallel
machines.  

In contrast, it remains a  notorious open problem in scheduling theory to
settle the approximability in the unrelated machine model (see e.g. ``open
problem 8'' in~\cite{schuurman1999polynomial}).  First, Schulz and
Skutella~\cite{schulz2002scheduling} and independently Chudak~\cite{Chudak99} came up
with  $(3/2+\epsilon)$-approximation algorithms, employing a time-indexed LP
relaxation for the problem. Shortly thereafter, the approximation guarantee was
improved to $3/2$ by Skutella~\cite{Skutella01} and Sethuraman and
Squillante~\cite{SethuramanS99} using a clever convex quadratic programming
relaxation. All these results relied on designing a convex relaxation and then
applying \emph{independent} randomized rounding with the marginal probabilities
that were returned by the convex relaxation solution. The analysis of these
algorithms is in fact tight: it is not hard to see that any algorithm using
independent randomized rounding  cannot achieve a better approximation
guarantee than $3/2$.  Recently, Bansal et al.~\cite{BansalSS16} overcame this barrier by
designing a randomized rounding scheme that, informally, enhances independent randomized
rounding by introducing strong negative correlation properties.  Their
techniques yield a~$(3/2- \varepsilon)$-approximation algorithm with respect to
either a semidefinite programming relaxation introduced by them or the
Configuration-LP relaxation introduced in~\cite{SviridenkoW13}. Their rounding and analysis
improve and build upon methods used previously  for independent
randomized rounding. So a natural question motivating this work is: can a
different rounding approach yield significant improvements of the approximation
guarantee?

\subsection{Our results}
Departing from previous rounding approaches, we propose to use the same elegant
rounding scheme for the weighted completion time objective as devised by Shmoys
and Tardos~\cite{ShmoysT93} for optimizing a linear function subject to makespan constraints on
unrelated machines. We give a \emph{tight} analysis which shows that this approach gives a significantly improved
approximation guarantee in the special case where the Smith ratios of all jobs
that can be processed on a machine are uniform: that is, we have
$p_{ij}\in\{ w_j,\infty\}$ for all $i\in\machines$ and $j\in\jobs$.\footnote{This restriction could be seen as close to the restricted assignment problem. However, we remark that all
our results also apply to the more general (but also more notation-heavy) case where the weight of a job may
also depend on the machine. A general version of our assumption then becomes $p_{ij} \in \{\alpha_i w_{ij}, \infty\}$ for some
machine-dependent $\alpha_i >0$. Our results apply to this version because our analysis will be done locally for each machine $i$, and
therefore we will only require that the Smith ratios be uniform for each machine separately.}

This restriction, which has not been studied previously, captures the natural notion that any unit of work (processing time) on a fixed machine has the same weight.
It corresponds to the class of instances where the order of jobs on a machine does not matter. 
Compared to another natural restriction of $R | | \sum w_j C_j$ -- namely, the \emph{unweighted} sum of completion times $R | | \sum C_j$ -- it is both computationally harder (in fact, whereas the unweighted version is polynomial-time solvable \cite{Horn73,BrunoCoffman74}, our problem inherits all the known hardness characteristics of the general weighted version: see Section~\ref{sec:lower_bounds_and_hardness})
and more intuitive (it is reasonable to expect that larger jobs have larger significance).
Despite the negative results, our main theorem indicates that we obtain far better understanding of this version than what is known for the general case.

To emphasize that we are considering the case where the weight of a job is
proportional to its processing time,  we refer to this problem as {\boldmath$R || \sum
p_j C_j$} (with $p_j$ as opposed to $w_j$). With this notation, our main result
can be stated as follows:
\begin{theorem}\label{thm:main}
  For any small $\varepsilon >0$, there exists a $\frac{1+\sqrt{2}}{2} + \varepsilon < 1.21$-approximation algorithm for  $R | | \sum p_j C_j$.
  Moreover, the analysis is tight: there exists an instance for which our algorithm returns a schedule with objective value at least $\frac{1+\sqrt{2}}{2} - \varepsilon$ times the optimum value.
\end{theorem}
We remark that the $\varepsilon$ in the approximation guarantee arises because
we can only solve the Configuration-LP relaxation (see Section\nobreakspace \ref {sec:prelim}) up to any desired
accuracy.

Interestingly enough, a similar problem (namely, scheduling jobs with uniform Smith ratios on \emph{identical parallel machines}) was studied by
Kawaguchi and Kyan \cite{DBLP:journals/siamcomp/KawaguchiK86}.\footnote{We would like to thank the anonymous \ifsoda\else SODA 2017 \fi reviewer for bringing this work
to our attention.} They achieve the same approximation ratio as we do, by using ideas in a similar direction as ours to
analyze a natural heuristic algorithm.

As we use the rounding algorithm by Shmoys and Tardos, a pleasant side-effect is
that our algorithm can also serve as a bi-criteria $(1+\sqrt{2})/2+\varepsilon$-approximation
for the $\sum p_j C_j$ objective and $2$-approximation for
the  makespan objective.\footnote{More precisely, given any makespan threshold $T > 0$, our algorithm will return a schedule with makespan at most $2T+\varepsilon$ and cost (i.e., sum of weighted completion times) within a factor $(1+\sqrt{2})/2+\varepsilon$ of the lowest-cost schedule among those with makespan at most $T$.} This bi-objective setting was previously studied by Kumar et al.~\cite{KumarMPS09}, who gave a bi-criteria $3/2$-approximation for the general weighted completion time objective and $2$-approximation for the makespan objective. 

Our main technical contribution is a tight analysis of the algorithm with
respect to the strong Configuration-LP relaxation. Configuration-LPs have been
used to design approximation algorithms for multiple important allocation
problems, often with great success; therefore, as first noted by Sviridenko and
Wiese~\cite{SviridenkoW13}, they constitute a promising direction to explore
in search for better algorithms for $R | | \sum w_j C_j$. We hope that our
analysis can give further insights as to how the Configuration-LP can be used to
achieve this.   

On a high level, our analysis proceeds as follows. A fractional solution to the
Configuration-LP  defines, for each machine, a local probability distribution of the set
of jobs (configuration) that will be processed by that machine. At the same time, the rounding
algorithm (naturally) produces a global distribution over such assignments, which
inherits certain constraints on the local
distribution for each machine. Therefore, focusing on a single machine, we will
compare the local input distribution (i.e., the one defined by the
Configuration-LP) to the {\em worst possible} (among ones satisfying said
constraints) local output distribution that could be returned by our randomized
rounding.\footnote{For example, if we consider all distributions that assign 2 jobs to a machine, each with probability $1/2$,
then the distribution which assigns either both jobs together or no job at all, each with probability $1/2$, is the worst possible
distribution, i.e., the one that maximizes the expected cost.} In order to analyze this ratio, we will have both 
distributions undergo a series of transformations that can only worsen
the guarantee, until we bring them into such a form that computing the exact
approximation ratio is possible. As the final form also naturally
corresponds to a scheduling instance, the tightness of our analysis follows
immediately.

\subsection{Lower Bounds and Hardness} \label{sec:lower_bounds_and_hardness}

\hspace*{\parindent}
All the known hardness features of the general problem $R || \sum w_j C_j$ transfer to our version $R | | \sum p_j C_j$.

First, it is implicit in the work of Skutella that both problems are APX-hard.\footnote{APX-hardness for $R | | \sum w_j C_j$ was first proved by Hoogeveen et al.~\cite{HoogeveenSW01}. Skutella~\cite[Section 7]{Skutella01} gives a different proof, where the reduction generates instances with all jobs having weights equal to processing times.}

Furthermore,
complementing the $(1+\sqrt{2})/2$ upper bound on the integrality gap of the Configuration-LP which follows from our algorithm, we have the
following lower bound, proved in Section\nobreakspace \ref{app:conflp-integrality-gap}:
\newcommand{\integralitygapconflp}{
	The integrality gap of the Configuration-LP for $R || \sum p_j C_j$ is at least $1.08$. 
}
\begin{theorem} \label{thm:gap}
	\integralitygapconflp
\end{theorem}

Finally, recall that
the $\frac 32$-approximation algorithms for the general problem $R | | \sum_j w_j C_j$ by Skutella~\cite{Skutella01} and by Sethuraman and
Squillante~\cite{SethuramanS99} are based on \emph{independent} randomized rounding of a fractional solution to a convex programming relaxation. It
was shown by Bansal et al.~\cite{BansalSS16} that this relaxation has an integrality gap of $\frac 32$ and that, moreover, no independent
randomized rounding algorithm can have an approximation ratio better than $\frac 32$. We note that both of these claims also apply to our
version. Indeed, the integrality gap example of Bansal et al. can be modified by having the job of size $k^2$ and weight $1$ instead have
size $k$ and weight $k$ (see~\cite[Claim~2.1]{BansalSS16}). For the second claim, their problematic instance already has only unit sizes and
weights. Thus, to get better than $\frac 32$-approximation to our version, one can use neither independent randomized rounding nor the
relaxation of~\cite{Skutella01,SethuramanS99}.

\subsection{Outline}
The paper is structured as follows. In Section\nobreakspace \ref {sec:prelim}, we start by defining
the Configuration-LP.  Then, in Section\nobreakspace \ref {sec:algorithm}, we describe the randomized
rounding algorithm by Shmoys and Tardos~\cite{ShmoysT93}
applied to our setting. We analyze it
Section\nobreakspace \ref {sec:analysis}.  Finally, in Section\nobreakspace \ref{app:conflp-integrality-gap} we present the proof of the
lower bound on the integrality gap.

%% file: prelim.tex
\section{The Configuration-LP} \label{sec:prelim}
As we are considering the case where $p_{ij} \in \{w_j, \infty\}$, we let $p_j
= w_j$. We further let $\jobs_i = \{j\in \jobs : p_{ij} = p_j\}$ denote the set
of jobs that can be assigned to machine $i\in \machines$. 

For intuition, consider an optimal schedule of the considered instance. Observe
that the schedule partitions the jobs into $|\machines|$ disjoint sets $\jobs = C_1 \cup C_2 \cup \dots \cup C_{|\machines|}$, where the jobs in $C_i$ are
scheduled on machine $i$ (so $C_i \subseteq \jobs_i$).  As described in the
introduction, the jobs  in $C_i$ are scheduled in an optimal way on machine $i$
by using a Smith ordering and, since we are considering the case where all jobs have the same Smith ratio, any ordering is optimal. The cost of scheduling the set of jobs $C_i$ on machine $i$ can therefore be written as 
\begin{align*}
  \cost(C_i)=\sum\limits_{j\in C_i} p^2_j+\sum\limits_{j\neq j'\in C_i}\frac{p_jp_{j'}}{2}.
\end{align*}
To see this, note that if we pick a random schedule/permutation of $C_i$, then
the expected completion time of job $j$ is $p_j  + \sum_{j' \neq j\in C_i}
\frac{p_{j'}}{2}$. The total cost of the considered schedule is
$\sum_{i\in \machines} \cost(C_i)$.

The Configuration-LP models, for each machine, the decision of which
configuration (set of jobs) that machine should process. Formally, we have
a variable $y_{iC}$ for each machine $i\in \machines$ and each configuration
$C\subseteq \jobs_i$ of jobs. The intended meaning of $y_{iC}$ is that it takes value $1$ if
$C$ is the set of jobs that machine $i$ processes, and it takes value $0$ otherwise. The constraints of
a solution are that each machine should
process at most one configuration and that each job should be processed exactly once.  The Configuration-LP can be compactly stated as follows: 
\[
\begin{array}{rrl}
  \min & \multicolumn{2}{l}{\quad  \displaystyle \sum_{i\in \machines}\sum_{C\subseteq \jobs_i}  y_{iC} \cost(C)}\\[7mm]
  \text{s.t.} & \displaystyle\sum_{C\subseteq \jobs_i} y_{iC} \leq 1& \quad \forall i\in \machines\,, \\[6mm]
& \displaystyle\sum_{i \in \machines} \displaystyle\sum_{C\subseteq \jobs_i: j\in C} y_{iC}= 1& \quad \forall j\in
\jobs\,,\\[3mm]
& \displaystyle y_{iC}\geq 0 &\quad \forall i\in \machines, \ C\subseteq \jobs_i\,.
\end{array}
\]
This linear program has an exponential number of variables and it is therefore
non-trivial to solve; however, Sviridenko and Wiese \cite{SviridenkoW13} showed
that, for any $\varepsilon >0$, there exists a polynomial-time algorithm that gives
a feasible solution to the relaxation whose cost is at most a factor $(1+\varepsilon)$ more
than the optimum.   Hence, the
Configuration-LP becomes a powerful tool that we use  to design a good
approximation algorithm for our problem.

%% file: alg.tex
\section{The Rounding Algorithm}
\label{sec:algorithm}
Here, we describe our approximation algorithm. We will analyze it in the next
section, yielding Theorem~\ref{thm:main}.  The first step of our algorithm is
to solve the Configuration-LP (approximately) to obtain a fractional
solution $\ys$. We then round this solution in order  to retrieve an integral
assignment of jobs to machines. The rounding algorithm that we employ is
the same as that used by Shmoys and Tardos~\cite{ShmoysT93}
(albeit applied to a fractional solution to the Configuration-LP, instead of the
so-called Assignment-LP). For completeness, we describe the
rounding scheme in Algorithm~\ref{alg:round}; see also Figure\nobreakspace \ref {fig:algorithm}.

\begin{algorithm2e*}[t!]
\SetKwInOut{Input}{Input}\SetKwInOut{Output}{Output}
\Input{Solution $\ys$ to the Configuration-LP}
\Output{Assignment of jobs to machines} 

\vspace{2mm}
1) Define $x\in \mathbb{R}^{\machines \times \jobs}$ as follows: $x_{ij}=\sum\limits_{C\subseteq \jobs_i:j\in C}\ys_{iC}$.

 2) Let $G=(U\cup V,E)$ be the complete bipartite graph where 
  \begin{itemize}\itemsep1mm
   \item the right-hand side consists of one vertex for each job $j$, i.e., $V=\{v_j:j\in\jobs\}$,
   \item the left-hand side consists of $\lceil \sum\limits_{j\in\jobs} x_{ij}\rceil$ vertices for each machine $i$, i.e.,\newline
   $U=\bigcup\limits_{i\in\machines}\{u_{i,t}:1\leq t\leq \lceil \sum\limits_{j\in\jobs} x_{ij}\rceil\}$.
  \end{itemize}

  \vspace{-2mm}
  3) Define a fractional solution $z$ to the bipartite matching LP for $G$ (initially set to $z=\bf{0}$) by repeating the following
procedure for every machine $i\in\machines$:
  \begin{itemize}\itemsep1mm
   \item Let $k=\lceil \sum_{j\in\jobs} x_{ij}\rceil$, and let $t$ be a variable originally set to 1.
   \item Iterate over all $j\in\jobs$ in \emph{non-increasing order}  in terms of $p_j$:
   \begin{itemize}
     \item[] If $x_{ij}+\sum\limits_{j'\in\jobs}z_{u_{i,t} v_{j'}} \le 1$, then set $z_{u_{i,t} v_j}=x_{ij}$.
     
     \item[] Else, set $z_{u_{i,t} v_j}=1-\sum\limits_{j'\in\jobs}z_{u_{i,t} v_{j'}}$,
     increment $t$, and set $z_{u_{i,t} v_j}=x_{ij}-z_{u_{i,t-1} v_j}$.
   \end{itemize}
  \end{itemize}
\vspace{-2mm}
4) Decompose $z$ into a convex combination of integral matchings $z=\sum_t \lambda_t z_t$ and sample one integral matching $z^*$ by
choosing the matching $z_t$ with probability $\lambda_t$.

\vspace{1mm}
5) Schedule $j\in\jobs$ on $i\in\machines$ iff $z^*_{u_{i,t} v_j}=1$ for some $1\leq t\leq \lceil \sum\limits_{j\in\jobs} x_{ij}\rceil$.
\caption{Randomized rounding} \label{alg:round}
\end{algorithm2e*}

The first step is to define $x_{ij} = \sum_{C \subseteq \jobs_i:
j\in C} \ys_{iC}$. Intuitively, $x_{ij}$ denotes the marginal probability that
job $j$ should be assigned to machine $i$, according to $\ys$.  Note that, by the
constraint that $\ys$ assigns each job once (fractionally), we have $\sum_{i\in \machines}
x_{ij} =1$ for each job $j\in \jobs$.

In the next steps, we round the fractional solution randomly so as to satisfy
these marginals, i.e., so that the probability that job $j$ is assigned to $i$ is
$x_{ij}$. In addition, the number of jobs assigned to a machine $i$ will
closely match the expectation $\sum_{j\in \jobs} x_{ij}$: our rounding will assign either $\lfloor
\sum_{j\in\jobs} x_{ij}\rfloor$ or $\lceil
\sum_{j\in\jobs} x_{ij}\rceil$  jobs to machine $i$.
This is enforced by creating  $\lceil\sum_{j\in\jobs} x_{ij}\rceil$ ``buckets''
for each machine $i$, and  then matching the jobs to these buckets.  More
formally, this is modeled by the complete bipartite graph $G=(U \cup V, E)$
constructed in Step 2 of Algorithm~\ref{alg:round}, where vertex $u_{i,t}\in U$
corresponds to the $t$-th bucket of machine $i$.


Observe that any integral matching in $G$ that matches all the ``job'' vertices
in $V$ naturally corresponds to an assignment of jobs to machines.  
Now, Step~3 prescribes a distribution on such matchings by defining
a~fractional matching $z$. The procedure is as follows: for each machine $i$,
we iterate over the jobs $j\in \jobs$ in \emph{non-increasing order} in terms of
their size, and we insert items of size $x_{ij}$ into the first bucket until
adding the next item would cause the bucket to become full; then, we split that
item between the first bucket and the second, and we proceed likewise for all jobs
until we fill up all buckets (except possibly for the last bucket).
%
%
%
%
%
Having completed this process for all machines $i$, we end up with a fractional matching $z$ in $G$ with the following properties:
\begin{itemize} \setlength{\itemsep}{0pt}
  \item Every ``job'' vertex $v_j\in V$ is fully matched in $z$, i.e., $\sum_{i, t} z_{u_{i,t} v_j} = 1$.
  \item For every ``bucket'' vertex $u_{i,t} \in U$, we have $\sum_j z_{u_{i,t} v_j} \leq 1$, with equality if $t< \lceil \sum_{j} x_{ij} \rceil$.
 \item The fractional matching preserves the marginals, i.e., $x_{ij}=\sum_{t}z_{u_{i,t}v_j}$ for all $j\in\jobs$ and $i\in\machines$.
 \item  We have the following bucket structure: if $z_{u_{i,t} v_j} > 0$ and $z_{u_{i,t'} v_{j'}} > 0$ with $t' > t$, then $p_j \geq p_{j'}$.
\end{itemize}
The last property follows because Step~3 considered the jobs in non-increasing order of their processing times; this will be important in the analysis (see the last property of Fact~\ref{fact:output_is_nice}).

Now, we want to randomly select a matching for $G$ which satisfies the
marginals of $z$ (remember that such a matching corresponds to an assignment of
all the jobs to machines). We know that the bipartite matching LP is integral
and that $z$ is a feasible solution for the bipartite matching LP of $G$;
therefore, using an algorithmic version of Carathéodory's theorem (see e.g. Theorem 6.5.11 in
\cite{GroetschelLovaszSchrijver1993}), we can decompose $z$ into a convex combination $z=\sum_t \lambda_t
z_t$ of polynomially many integral matchings, and sample the matching $z_t$ with probability
$\lambda_t$. Then, if $z^*$ is the matching we have sampled, we simply assign job $j$
to machine $i$ iff $z^*_{u_{i,t}v_j}=1$ for some $t$. Since
$x_{ij}=\sum_{t}z_{u_{i,t}v_j}$ and $\sum_{i\in\machines} x_{ij}=1$ for all jobs
$j$, $z^*$ will match all ``job'' vertices.\footnote{We remark that this is the only part of the algorithm that employs randomness; in fact, we
can derandomize the algorithm by choosing the matching $z_k$ that minimizes the cost of the resulting
assignment.}
The above steps are described in Steps 4 and 5 of Algorithm \ref{alg:round}.

The entire rounding algorithm is depicted in Figure\nobreakspace \ref {fig:algorithm}.

\begin{figure*}[t!]
\begin{minipage}[t]{\linewidth}
    \begin{tikzpicture}[scale=0.6]
      \begin{scope}[shift={(6.5,0)}]
      \draw[draw=gray] (0,0) edge[dashed] (0,-1.8);
      \draw[draw=gray] (3,0) edge[dashed] (3,-1.8);
      \draw[draw=gray] (6,0) edge[dashed] (6,-1.8);
      \draw[draw=gray] (9,0) edge[dashed] (9,-1.8);
      \node at (1.5, -1.5) {\footnotesize 1st bucket};
      \node at (4.5, -1.5) {\footnotesize 2nd bucket};
      \node at (7.5, -1.5) {\footnotesize 3rd bucket};

      \draw[fill=black!20!white] (0,0) rectangle (2,2.2);
      \node at (1, -0.5) {\small $\nicefrac{2}{3}$};
      \draw[fill=black!50!white](2,0) rectangle (3,1.8);
      \node at (2.5, -0.5) {\small$\nicefrac{1}{3}$};
      \draw[fill=black!70!white] (3,0) rectangle (5,1.4);
      \node at (4, -0.5) {\small$\nicefrac{2}{3}$};
      \draw[pattern=north east lines]  (5,0) rectangle (7,0.8);
      \node at (6, -0.5) {\small$\nicefrac{2}{3}$};
      \draw (7,0) rectangle (8,0.3);
      \node at (7.5, -0.5) {\small$\nicefrac{1}{3}$};
      
              \node at (10.5,1.5) {$\Rightarrow$};
      
        \node at (4.5, 5.5) {\small \begin{minipage}{4.5cm} Bucketing of machine $i^\star$ \end{minipage}};
      \end{scope}

      %
      \begin{scope}[xshift=0cm]
        
        \node at (0, -0.5) {\small$0$};
        \draw[fill=black!20!white] (0,0) rectangle (1,2.2);
        \draw[fill=black!70!white](0,2.2) rectangle (1,3.6);
        \node at (1, -0.5) {\small$\nicefrac{1}{3}$};
        
        \draw[fill=black!20!white] (1,0) rectangle (2,2.2);
        \draw[pattern=north east lines]  (1,2.2) rectangle (2,3);
        \draw (1,3) rectangle (2,3.3);
        \node at (2, -0.5) {\small$\nicefrac{2}{3}$};
        
        \draw[fill=black!50!white](2,0) rectangle (3,1.8);
        \draw[fill=black!70!white] (2,1.8) rectangle (3,3.2);
        \draw[pattern=north east lines]  (2,3.2) rectangle (3,4);

        \draw (-0.4, 0) edge[-] (3.4, 0);
        \node at (3, -0.5) {\small$1$};
        
        \node at (4.75,1.5) {$\Rightarrow$};

        \node at (1.5, 5.5) {\small \begin{minipage}{4.5cm} Input distribution on \\ configurations (patterns) \\of machine $\istar$ ($\yin \leadsto g$) \end{minipage}};
      \end{scope}
      
            \begin{scope}[shift={(17,0)}]

            \node at (4.5, 5.5) {\small \begin{minipage}{4.5cm} Fractional matching \end{minipage}};
            \end{scope}
            
            \begin{scope}[shift={(11,-18)}]
            
            \node at (4.5, 5.5) {\small \begin{minipage}{4.5cm} Combination of matchings \end{minipage}};
            \end{scope}
            
                        \begin{scope}[shift={(21,-2.5)}]
                        
                        \draw (-2,6)  -- (2,6) node[midway,above] {\tiny$2/3$};
                        \draw (-2,6)  -- (2,4.75) node[midway,above] {\tiny$1/3$};
                        \draw (-2,3.5)  -- (2,3.5) node[midway,above] {\tiny$2/3$};
                        \draw (-2,3.5)  -- (2,2.25) node[midway,above] {\tiny$1/3$};
                        \draw (-2,1)  -- (2,2.25) node[midway,above] {\tiny$1/3$};
                        \draw (-2,1)  -- (2,1) node[midway,above] {\tiny$1/3$};
                        
                        \filldraw (-2,1) circle (0.2);
                        \filldraw (-2,3.5) circle (0.2);
                        \filldraw (-2,6) circle (0.2);
                        
                        \filldraw[fill=white,draw=none](1.8,2.05) rectangle (2.2,2.45);
                        \filldraw[fill=white,draw=none]  (1.8,2.05) rectangle (2.2,2.45);
                        \filldraw[fill=white,draw=none] (1.8,.8) rectangle (2.2,1.2);

                        \draw[fill=black!20!white] (1.8,5.8) rectangle (2.2,6.2);
                        \draw[fill=black!50!white](1.8,4.55) rectangle (2.2,4.95);
                        \draw[fill=black!70!white] (1.8,3.3) rectangle (2.2,3.7);
                        \draw[pattern=north east lines]  (1.8,2.05) rectangle (2.2,2.45);
                        \draw (1.8,.8) rectangle (2.2,1.2);
                        
              \node at (0,-1) {$\Downarrow$};
                        
                        \end{scope}
                        
                        \begin{scope}[shift={(21,-12)}]
                        \draw (-2,6)  -- (2,4.75);
                        \draw (-2,3.5)  -- (2,2.25);
                        
                        \filldraw (-2,1) circle (0.2);
                        \filldraw (-2,3.5) circle (0.2);
                        \filldraw (-2,6) circle (0.2);
                        
                        \filldraw[fill=white,draw=none](1.8,2.05) rectangle (2.2,2.45);
                        \filldraw[fill=white,draw=none]  (1.8,2.05) rectangle (2.2,2.45);
                        \filldraw[fill=white,draw=none] (1.8,.8) rectangle (2.2,1.2);

                        \draw[fill=black!20!white] (1.8,5.8) rectangle (2.2,6.2);
                        \draw[fill=black!50!white](1.8,4.55) rectangle (2.2,4.95);
                        \draw[fill=black!70!white] (1.8,3.3) rectangle (2.2,3.7);
                        \draw[pattern=north east lines]  (1.8,2.05) rectangle (2.2,2.45);
                        \draw (1.8,.8) rectangle (2.2,1.2);
                        
                        \node at (0,7) {$\nicefrac{1}{3} \ \times$};
                        \end{scope}
                        
                        \begin{scope}[shift={(15,-12)}]
                         
                         \draw (-2,6)  -- (2,6);
                         \draw (-2,3.5)  -- (2,3.5);
                         \draw (-2,1)  -- (2,1);
                         
                         \filldraw (-2,1) circle (0.2);
                         \filldraw (-2,3.5) circle (0.2);
                         \filldraw (-2,6) circle (0.2);
                         
                         \filldraw[fill=white,draw=none](1.8,2.05) rectangle (2.2,2.45);
                         \filldraw[fill=white,draw=none]  (1.8,2.05) rectangle (2.2,2.45);
                         \filldraw[fill=white,draw=none] (1.8,.8) rectangle (2.2,1.2);

                         \draw[fill=black!20!white] (1.8,5.8) rectangle (2.2,6.2);
                         \draw[fill=black!50!white](1.8,4.55) rectangle (2.2,4.95);
                         \draw[fill=black!70!white] (1.8,3.3) rectangle (2.2,3.7);
                         \draw[pattern=north east lines]  (1.8,2.05) rectangle (2.2,2.45);
                         \draw (1.8,.8) rectangle (2.2,1.2);
                        \node at (0,7) {$\nicefrac{1}{3} \ \times$};
                        \node at (3,7) {$+$};
                        \end{scope}
                        
                        \begin{scope}[shift={(9,-12)}]
                        
                        \draw (-2,6)  -- (2,6);
                        \draw (-2,3.5)  -- (2,3.5);
                        \draw (-2,1)  -- (2,2.25);
                        
                        \filldraw (-2,1) circle (0.2);
                        \filldraw (-2,3.5) circle (0.2);
                        \filldraw (-2,6) circle (0.2);
                        
                        \filldraw[fill=white,draw=none](1.8,2.05) rectangle (2.2,2.45);
                        \filldraw[fill=white,draw=none]  (1.8,2.05) rectangle (2.2,2.45);
                        \filldraw[fill=white,draw=none] (1.8,.8) rectangle (2.2,1.2);

                        \draw[fill=black!20!white] (1.8,5.8) rectangle (2.2,6.2);
                        \draw[fill=black!50!white](1.8,4.55) rectangle (2.2,4.95);
                        \draw[fill=black!70!white] (1.8,3.3) rectangle (2.2,3.7);
                        \draw[pattern=north east lines]  (1.8,2.05) rectangle (2.2,2.45);
                        \draw (1.8,.8) rectangle (2.2,1.2);
                        \node at (0,7) {$\nicefrac{1}{3} \ \times$};
                        \node at (3,7) {$+$};
                        \node at (-4.25,3.5) {$\Leftarrow$};
                        \end{scope}
                        
      \begin{scope}[shift={(0,-10)}]
      
      \node at (0, -0.5) {\small$0$};
      \draw[fill=black!20!white] (0,0) rectangle (1,2.2);
      \draw[fill=black!70!white](0,2.2) rectangle (1,3.6);
      \draw[pattern=north east lines]  (0,3.6) rectangle (1,4.4);
      \node at (1, -0.5) {\small$\nicefrac{1}{3}$};
      
      \draw[fill=black!20!white] (1,0) rectangle (2,2.2);
      \draw[fill=black!70!white]  (1,2.2) rectangle (2,3.6);
      \draw (1,3.6) rectangle (2,3.9);
      \node at (2, -0.5) {\small$\nicefrac{2}{3}$};
      
      \draw[fill=black!50!white](2,0) rectangle (3,1.8);
      \draw[pattern=north east lines] (2,1.8) rectangle (3,2.6);
      
      \draw (-0.4, 0) edge[-] (3.4, 0);
      \node at (3, -0.5) {\small$1$};

      \node at (1.5, -2.5) {\small \begin{minipage}{4.5cm} Output distribution on\\configurations (patterns)\\ of machine $\istar$ ($\yout \leadsto f$) \end{minipage}};
      \end{scope}

       \end{tikzpicture}
    \end{minipage}%
    \caption{A sample execution of our rounding algorithm, restricted to a single machine $\istar$. Jobs are represented by a rectangle; its height is the job's processing time and its width is its fractional assignment to $\istar$.  
    Starting from an input distribution over configurations for $\istar$, we extract the fractional assignment of each job to $\istar$, we create a bipartite graph consisting of 3 copies of $\istar$ and the jobs that are fractionally assigned to it, and then we connect the jobs to the copies of $\istar$ by iterating through the jobs in non-increasing order of $p_j$. Finally, we decompose the resulting fractional matching into a convex combination of integral matchings and we sample one of them.
	The shown output distribution is a worst-case distribution in the sense of~Section\nobreakspace \ref {sec:worstcasedist}: it maximizes the variance of makespan, subject to the marginal probabilities and the bucket structure enforced by the algorithm.
}
    
    \label{fig:algorithm}

\end{figure*}
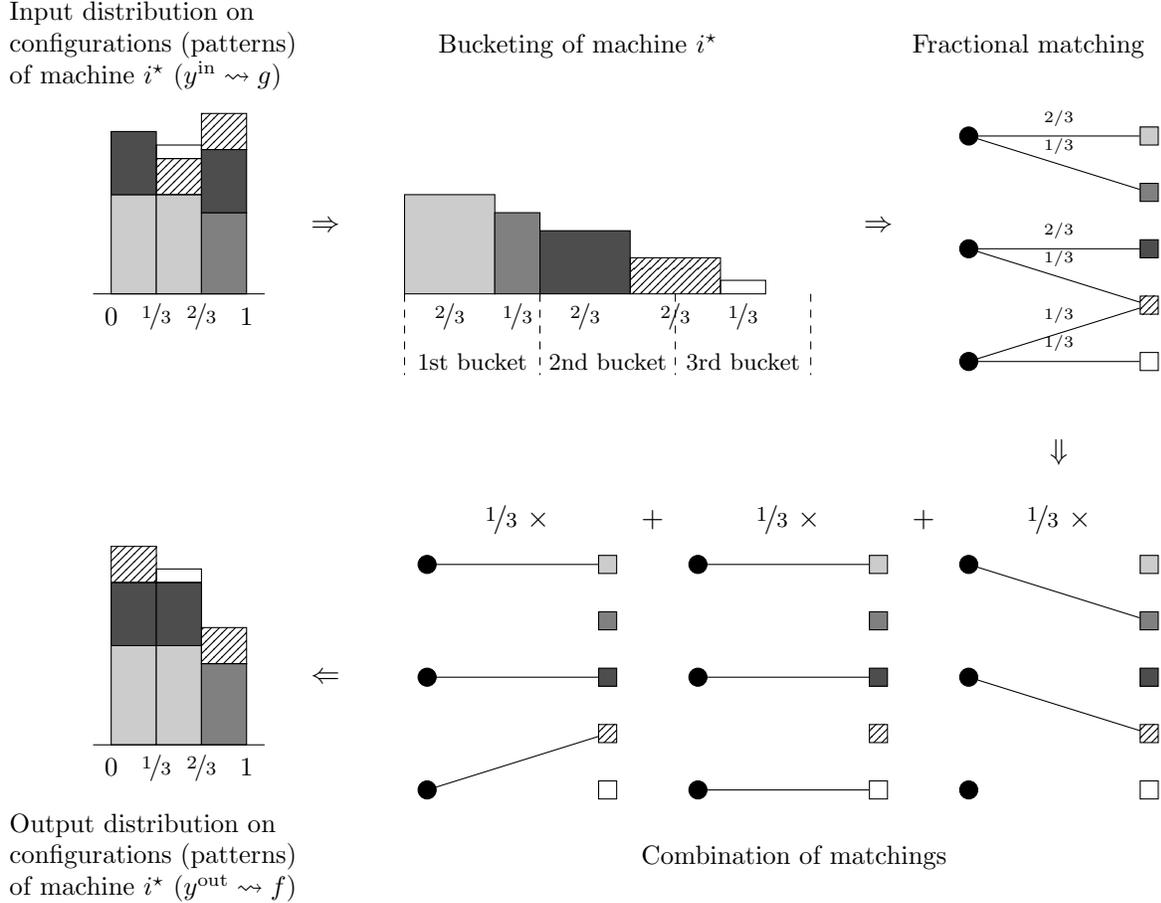

%% file: analysis.tex
\section{Analysis}
\label{sec:analysis}

Throughout the analysis, we fix a single machine $\istar \in \machines$. We will show
that the expected cost of our algorithm on this machine is at most $\frac{1+\sqrt{2}}{2}$ times
the cost of the Configuration-LP on this machine. This clearly implies (by linearity of expectation) that the
expected cost of the produced solution (on all machines) is at most $\frac{1+\sqrt{2}}{2}$
times the cost of the LP solution (which is in turn within a factor $(1+\eps)$ of the fractional optimum).

Let $C_1, C_2, ..., C_{2^\jobs}$ be all possible configurations sorted by decreasing cost, i.e., $\cost(C_1) \ge \ldots \ge \cost(C_{2^\jobs})$.
To simplify notation, in this section we let $\jobs$ denote the set of jobs that can be processed on machine $\istar$ (i.e., $\jobs_{\istar}$).

Recall that the solution $\ys$ to the Configuration-LP gives us an \emph{input} distribution on configurations assigned to machine $\istar$,
i.e., it gives us a vector $\yin \in [0,1]^{2^\jobs}$ such that $\sum_i \yin_i = 1$. With this notation, we can write the cost of the
Configuration-LP on machine $\istar$ as \[ \sum_i \yin_i \cost(C_i). \]

In order to compare this expression with the expected cost of our algorithm on machine $\istar$, we observe that our rounding algorithm also gives a distribution on
configurations. We denote this \emph{output} distribution by $\yout \in [0,1]^{2^\jobs}$ (where $\sum_i \yout_i = 1$). Hence, the expected cost of our algorithm on machine $\istar$ is \[ \sum_i \yout_i \cost(C_i). \]

The result of this section that implies the approximation guarantee of~Theorem\nobreakspace \ref {thm:main} can now be stated as follows.
\begin{theorem} \label{thm:maintech}
  We have \[ \frac{\sum_i \yout_i \cost(C_i)}{\sum_i \yin_i \cost(C_i)} \le \frac{1+\sqrt{2}}{2}. \]
\end{theorem}

Our strategy
for bounding this ratio
is, broadly, to work on this pair of distributions by transforming it to another pair of distributions of special form, whose ratio we will be able to bound. We transform the pair in such a way that the ratio can only increase. In other words, we prove that no pair of distributions has a worse ratio than a certain worst-case kind of pair, and we bound the ratio in that worst case.

After these transformations, our pair of distributions may no longer correspond to the original scheduling problem instance, so it will be convenient for us to work with a more abstract notion that we define now.

\subsection{Compatible Function Pairs}

Given a distribution $y \in [0,1]^{2^\jobs}$ with $\sum_i y_i = 1$, we can build a corresponding function $f$ from  $[0,1)$ to multisets of positive numbers as follows: define $f(x)$ for $x \in [0,y_1)$ to be the multiset of processing times of jobs in $C_1$,
$f(x)$ for $x \in [y_1, y_1 + y_2)$ to be the multiset of processing times of jobs in $C_2$,
and so on.\footnote{Recall that $C_1, C_2, ...$ are sorted by non-increasing cost. Thus $f$ can be thought of as a quantile function (inverse cumulative distribution function) of the distribution $y$, except in reverse order (i.e., $f(1-x)$ is a quantile function).}
If we do this for both $\yout$ -- obtaining a function $f$ -- and $\yin$ -- obtaining a function $g$ (see~Figure\nobreakspace \ref {fig:algorithm} for an illustration of $f$ and $g$), we will have produced a \emph{function pair}:

\begin{definition}
	A \emph{function pair} is a pair $(f,g)$ of stepwise-constant functions from the interval $[0,1)$ to multisets of positive numbers.
	We will call these multisets \emph{patterns} and the numbers they contain \emph{elements} (or \emph{processing times}).
\end{definition}
\paragraph{Notation.}
If $f$ is such a function, define:
\begin{itemize} 
	\item $f_1 : [0,1) \to \bR_+$ as the maximum element: $f_1(x) = \max f(x)$ (set $0$ if $f(x) = \emptyset$),
	\item $\size_f : [0,1) \to \bR_+$ as $\size_f(x) = \size(f(x))$, where
		\[ \size(f(x)) = \sum_{p \in f(x)} p, \]
	\item $f_r$ as the total size of the multiset after the removal of the maximum: $f_r(x) = \size_f(x) - f_1(x)$,
	\item $\cost(f) = \int_0^1 \cost(f(x)) \, dx$ as the fractional (expected) cost, where
		\[ \cost(f(x)) = \sum_{p \in f(x)} p \cdot \rb{ \sum_{q \in f(x), q \preceq  p} q } \]
		for an arbitrary linear order $\preceq$ on $f(x)$.\footnote{This expression does not depend on $\preceq$, since the Smith ratios are
uniform, and it is equal to the cost of a configuration giving rise to $f(x)$.}
\end{itemize}

Function pairs we work with will have special properties that follow from the algorithm. We argue about them in~\protect \MakeUppercase {F}act\nobreakspace \ref {fact:output_is_nice}. One such property comes from our algorithm preserving the marginal probabilities of jobs:
\begin{definition}
	We say that a function pair $(f,g)$ is a \emph{compatible function pair (CFP)} if the fractional number of occurences of any element is the same in $f$ and in $g$.\footnote{Formally, for each $p > 0$ we \ifsoda\else have: \fi
	$ \int_0^1 \text{multiplicity of $p$ in $f(x)$ } dx = \int_0^1 \text{multiplicity of $p$ in $g(x)$ } dx$.}
\end{definition}

\begin{fact} \label{fact:output_is_nice}
	Let $(f,g)$ be a function pair obtained from $(\yout,\yin)$ as described above. Then:
	\begin{itemize} 
		\item $(f,g)$ is a CFP.
		\item $\cost(f) = \sum_i \yout_i \cost(C_i)$ and $\cost(g) = \sum_i \yin_i \cost(C_i)$.
		\item $f$ has the following \emph{bucket structure}: for any two patterns $P$ and $Q$ in the image of $f$ and for any $i$, the $i$-th largest element of $P$ is no smaller than the $(i+1)$-th largest element of $Q$.
	\end{itemize}
\end{fact}
\begin{proof}
	That $(f,g)$ is a CFP follows because our algorithm satisfies the marginals of the involved jobs. Namely, we know that for each job $j \in \jobs$, both distributions $\yin$ and $\yout$ have the machine $\istar$ process a~fraction $x_{\istar j} = \sum_{C \ni j} \ys_{\istar C}$ of this job (where $\ys$ is the global Configuration-LP solution). For any $p > 0$, summing this up over all jobs $j \in \jobs$ with processing time $p_j = p$ gives the compatibility condition.
	
	The equalities of costs are clear from the definition of $\cost(f)$ and $\cost(g)$.
	
	For the bucket structure of $f$, recall the algorithm: a matching is found between jobs and buckets, in which each bucket is matched with a job (except potentially for the last bucket of each machine). For any pattern $P$ in the image of $f$ and for any $i$, the $i$-th processing time in $P$ is drawn from the $i$-th bucket that was constructed by our algorithm. Moreover, all processing times in the $i$-th bucket are no smaller than those in the $(i+1)$-th bucket, because the algorithm orders jobs non-increasingly by processing times. (See~Figure\nobreakspace \ref {fig:algorithm} for an illustration of this process and of a function $f$ satisfying this bucket structure.)
\end{proof}

This was the last point in the analysis where we reasoned about how the algorithm rounds the LP solution.
From now on, we will think about elements, patterns and CFPs rather than jobs and configurations.

To prove~Theorem\nobreakspace \ref {thm:maintech}, we need to show that $\frac{\cost(f)}{\cost(g)} \le \frac{1+\sqrt{2}}{2}$. As indicated above, we will do this by proving that there is another CFP $(f',g')$ with special properties and such that $\frac{\cost(f)}{\cost(g)} \le \frac{\cost(f')}{\cost(g')}$. We will actually construct a series of such CFPs in a series of lemmas, obtaining more and more desirable properties, until we can bound the ratio. Our final objective is a CFP like the pair $(f',g')$ depicted in~Figure\nobreakspace \ref {fig:idealdistro}. 

\subsection{The Worst-Case Output} \label{sec:worstcasedist}

As a first step, we look at how costly an output distribution of our algorithm can be (while still satisfying the aforementioned bucket structure and the marginal probabilities, i.e., the compatibility condition). Intuitively, the maximum-cost $f$ is going to maximize the variance of the total processing time, which means that larger-size patterns should select larger processing times from each bucket. (See~Figure\nobreakspace \ref {fig:algorithm} for an illustration and the proof of~Lemma\nobreakspace \ref {lem:worst_case_output} for details.) From this, we extract that the largest processing time in a pattern (the function $f_1$) should be non-increasing, and this should also hold for the second-largest processing time, the third-largest, and so on. This implies the
following properties:

\newcommand{\worstcaseoutputstatement}{
	If $(f,g)$ is a CFP where $f$ has the bucket structure described in \protect \MakeUppercase {F}act\nobreakspace \ref {fact:output_is_nice}, then there exists another CFP $(f',g)$ such that
	$\frac{\cost(f)}{\cost(g)} \le \frac{\cost(f')}{\cost(g)}$
	and
	the functions $f_1'$ and $f_r'$ are non-increasing
	and
	$\size(f'(1)) \ge f_r'(0)$.
}
\begin{lemma} \label{lem:worst_case_output}
	\worstcaseoutputstatement
\end{lemma}
The proof is a simple swapping argument.
\begin{proof}
	For $i = 1, 2, ...$, let $f_i(x)$ always denote the $i$-th largest element of $f(x)$.
	As suggested above, we will make sure that for each $i$, the function $f_i$ is non-increasing.
	
	Namely, we repeat the following procedure: as long as there exist $x$, $y$ and $i$ such that $\size(f(x)) > \size(f(y))$ but $f_i(x) < f_i(y)$, swap the $i$-th largest elements in $f(x)$ and $f(y)$.\footnote{Formally, choose $\tau > 0$ such that $f$ is constant on $[x,x+\tau)$ and on $[y,y+\tau)$ and perform the swap in these patterns.}
	
	Once this is no longer possible, we finish by ``sorting'' $f$ so as to make $\size_f$ non-increasing.
	
	Let us verify that once this routine finishes, yielding the function $f'$, we have the desired properties:
	\begin{itemize}
		\item The function $f_1'$ is non-increasing.
		\item The same holds for the function $f_r'$, since $f_r' = f_2' + f_3' + ...$ and each $f_i'$ is non-increasing.
		\item The procedure maintains the bucket structure, which implies that $f_i'(x) \ge f_{i+1}'(y)$ for all $i$, $x$ and $y$. Thus
		\[ f'(1) = f_1'(1) + f_2'(1) + ... \ge f_2'(0) + f_3'(0) + ... = f_r'(0). \]
		\item It remains to show that $\cost(f') \ge \cost(f)$. Without loss of generality, assume there was only a single swap (as the sorting step is insignificant for the cost). For computing the cost of the involved patterns, we will think that the involved elements went last (since the order does not matter); let $R_x = \size(f(x)) - f_i(x)$ and $R_y = \size(f(y)) - f_i(y)$ be the total sizes of the elements not involved.
		\ifsoda
\else
			Then $R_x \ge R_y$ and
			\begin{align*}
				\Delta \cost(f(x)) &= f_i(y) \rb{ R_x + f_i(y) } - f_i(x) \rb{ R_x + f_i(x) } \\
				                   &= \rb{ f_i(y) - f_i(x) } R_x + f_i(y)^2 - f_i(x)^2, \\
				\Delta \cost(f(y)) &= f_i(x) \rb{ R_y + f_i(x) } - f_i(y) \rb{ R_y + f_i(y) } \\
				                   &= \rb{ f_i(x) - f_i(y) } R_y + f_i(x)^2 - f_i(y)^2,
			\end{align*}
			thus
			\[ \Delta \cost(f(x)) + \Delta \cost(f(y)) = \rb{ f_i(y) - f_i(x) } \rb{ R_x - R_y } > 0. \]
		\fi
	\end{itemize}
\end{proof}

\subsection{Liquification}

One of the most important operations we will employ is called \emph{liquification}. It is the process of replacing an element (processing time) with many tiny elements of the same total size. These new elements will all have a size of $\eps$ and will be called \emph{liquid elements}. Elements of size larger than $\eps$ are called \emph{solid}. One should think that $\eps$ is arbitrarily small, much smaller than any $p_j$; we will usually work in the limit $\eps \to 0$.

The intuition behind applying this process to our pair is that in the ideal worst-case setting which we are moving towards, there are only elements of two sizes: large and infinitesimally small. We will keep a certain subset of the elements intact (to play the role of large elements), and liquify the rest in~Lemma\nobreakspace \ref {lem:introduce_m}.

Our main claim of this section is that replacing an element with smaller ones of the same total size (in both $f$ and $g$) can only increase the ratio of costs. Thus we are free to liquify elements in our analysis as long as we make sure to liquify the same amount of every element in both $f$ and $g$ ($f$ and $g$ remain compatible).

\newcommand{\liquificationcoststatement}{
	Let $(f,g)$ be a CFP and $p, p_1, p_2 > 0$ with $p = p_1 + p_2$. Suppose $(f',g')$ is a CFP obtained from $(f,g)$ by replacing $p$ by $p_1$ and $p_2$ in subsets of patterns in $f$ and $g$ of equal measures.\footnote{Formally, suppose $I_f, I_g \subseteq [0,1)$ are finite unions of disjoint intervals of equal total length such that all patterns $f(x)$ and $g(y)$ for $x \in I_f$, $y \in I_g$ contain $p$; in all these patterns, remove $p$ and add $p_1, p_2$.} Then $\frac{\cost(f)}{\cost(g)} \le \frac{\cost(f')}{\cost(g')}$.
}
\begin{fact} \label{fact:liquification_cost}
	\liquificationcoststatement
\end{fact}
\begin{proof}
	Consider a pattern $P$ in which $p$ was replaced. We calculate the change in cost $\Delta := \cost(P \setminus \{p\} \cup \{p_1, p_2\}) - \cost(P)$. As the order does not matter, the cost can be analyzed with $p$ (or $p_1, p_2$) being first, so
	\ifsoda
\else
		  \begin{align*}
		    \Delta = 
		    (p_1^2 + p_2 (p_1 + p_2)) - p^2 = (p_1^2 + p_1p_2 + p_2^2) - (p_1 + p_2)^2 = - p_1p_2 \le 0
		  \end{align*}
	  \fi
	  and it does not depend on the other elements in $P$.
	  Thus, if we make the replacement in a fraction $\tau$ of patterns, then we have
	\[
	\frac{\Cost{f'}}{\Cost{g'}}=\frac{\Cost{f} + \tau \Delta}{\Cost{g} + \tau \Delta} \geq \frac{\Cost{f}}{\Cost{g}}
	\]
	since $\frac{\cost(f)}{\cost(g)} \ge 1$ to begin with (otherwise we are done) and $\Delta \le 0$.
\end{proof}
By corollary, we can also replace an element of size $p$ with $p/\eps$ liquid elements (of size $\eps$ each).

\subsection{The Main Transformation}

In this section we describe the central transformation in our analysis. It uses liquification and rearranges elements in $f$ and $g$ so as to obtain two properties: that $f_r$ is constant and that $f_1 = g_1$.
(The process is explained in~Figure\nobreakspace \ref {fig:define_m}, and a resulting CFP is shown in the upper part of~Figure\nobreakspace \ref {fig:idealdistro}.)
This greatly simplifies the setting and brings us quite close to our ideal CFP (depicted in the lower part of~Figure\nobreakspace \ref {fig:idealdistro}).

\newcommand{\introducemstatement}{
	If $(f,g)$ is a CFP with $f_1$ and $f_r$ non-increasing and $\size(f(1)) \ge f_r(0)$, then there exists another CFP $(f',g')$ with $\frac{\cost(f)}{\cost(g)} \le \frac{\cost(f')}{\cost(g')}$ such that:
	\begin{enumerate}[(a)] 
		\item $f_r'$ is constant.
		\item $f_1' = g_1'$ and it is non-increasing.
		\item There exists $m \in [0,1]$ such that:
			\begin{itemize}
				\item for $x \in [0,m)$, $f'(x)$ has liquid elements and exactly one solid element,
				\item for $x \in [m,1)$, $f'(x)$ has only liquid elements.
			\end{itemize}
	\end{enumerate}
}
\begin{lemma} \label{lem:introduce_m}
	\introducemstatement
\end{lemma}
\newcommand{\introducemproofsketch}{
	Our procedure to obtain such a CFP has three stages:
	\begin{enumerate} 
		\item We liquify all elements except for the largest element of every pattern $f(x)$ for $x \in [0,m)$, for a certain
threshold $m$; at this point, $f_r(x)$ becomes essentially the size of liquid elements in $f(x)$ for $x \in [0,1)$ (since no pattern
contains two solid elements).
		\item We move liquid elements in $f$ from right to left, i.e., from $f(y)$ for $y \in [m,1)$ to $f(x)$ for $x \in [0,m)$
(which only increases the cost of $f$), so as to make $f_r$ constant.
		\item We rearrange elements in $g$ so as to satisfy the condition $f_1' = g_1'$.
	\end{enumerate}
}
\input{define_m}
\begin{proof} We begin with an overview of the proof. See also~Figure\nobreakspace \ref {fig:define_m}.
	
	\introducemproofsketch
	
	Now let us proceed to the details.
	First, we explain how to choose the threshold $m$. It is defined so as to ensure that after the liquification, patterns in $f$ have liquid elements of total size $f_r(0)$ on average. Thus we can make all $f_r(x)$ equal to $f_r(0)$. More importantly, we can do this by only moving the liquid elements ``higher'', so that the cost of $f$ only goes up. (See~Figure\nobreakspace \ref {fig:define_m} for an illustration of this move.)
	
	More precisely, $m \in [0,1]$ is chosen so that
	\begin{equation} \label{eq:define_m}
		\int_0^m f_r(0) - f_r(x) \, dx = \int_m^1 \size(f(x)) - f_r(0) \, dx
	\end{equation}
	(which implies $f_r(0) = \int_0^m f_r(x) \, dx + \int_m^1 \size(f(x)) \, dx$: the right-hand expression is the size of elements which will be liquified next). Such an $m$ is guaranteed to exist by nonnegativity of the functions under both integrals in~\eqref{eq:define_m}, which follows from our assumptions on $f$ (we have $\size(f(x)) \ge \size(f(1)) \ge f_r(0)$, because $\size_f = f_1 + f_r$ is non-increasing).
	
	Now we can perform the sequence of steps:
	
	\begin{enumerate}
	\item The liquification: for each element $p > 0$ which appears in a pattern $f(x) \setminus \{f_1(x)\}$ for $x \in [0,m)$ or in a pattern $f(x)$ for $x \in [m,1)$, we liquify all its such occurences, as well as their counterparts in $g$.\footnote{Formally, let $I_f = \{ x \in [0,m) : p \in f(x) \setminus \{f_1(x)\} \} \cup \{ x \in [m,1) : p \in f(x) \}$ and $I_g = \{ y \in [0,1): p \in g(y) \}$; by compatibility of $f$ and $g$, the measure of $I_g$ is no less than that of $I_f$. We liquify $p$ in $I_f$ and in a subset of $I_g$ of the right measure. (If there were patterns where $p$ appeared multiple times, we repeat.)} We have a bound on the ratio of costs by \protect \MakeUppercase {F}act\nobreakspace \ref {fact:liquification_cost}, and $f_1$ remains non-increasing.
	
	\item Rearranging liquid elements in $f$: while there exists $x \in [0,m)$ with $f_r(x) < f_r(0)$, find $y \in [m,1)$ with $f_r(y) > f_r(0)$ and move a liquid element from $f(y)$ to $f(x)$.\footnote{Formally, let $\tau > 0$ be such that $f_r(x') < f_r(0)$ for $x' \in [x,x+\tau)$ and $x + \tau \le m$ and $f_r(y') > f_r(0)$ for $y' \in [y,y+\tau)$. Move the liquid elements between these patterns $f(y')$ and $f(x')$.} (Since we want to make $f_r$ constant and equal to $f_r(0)$, we move elements from where $f_r$ is too large to where it is too small. Note that since we liquified all elements in $f(x)$ for $x \in [m,1)$, now $\size_f$ is almost equal to $f_r$ on $[m,1)$.) Once this process finishes, $f_r$ will be constant by the definition of $m$.\footnote{Since we operate in the limit $\eps \to 0$, we ignore issues such as $f_r$ being $\pm \eps$ off.} (See the right side of~Figure\nobreakspace \ref {fig:define_m}.) Note that $\size_f$ was non-increasing at the beginning of this step, and so we are always moving elements from
patterns of smaller total size to patterns of larger total size -- this can only increase the cost of $f$ and thus the ratio of costs increases. This step preserves $f_1$.

	\item Rearranging elements in $g$: at this point, since $f$ and $g$ are compatible, $g$ only has solid jobs that appear as $f_1(x)$ for $x \in [0,m)$, but they might be arranged differently in $g$ than in $f$. We want them to have the same arrangement (i.e., $f_1' = g_1'$) so that we can compare $f$ to $g$ more easily. As a first step, we make sure that every pattern in $g$ has at most one solid element. To this end, while there exists $x \in [0,1)$ such that $g(x)$ contains two or more solid elements, let $p > 0$ be one of them, and find $y \in [0,1)$ such that $g(y)$ contains only liquid elements\footnote{Such a $y$ exists because $f$ and $g$ are compatible: the total measure of patterns with solid elements in $f$ is $m$ and these patterns contain only one solid element each, so $g$ must have patterns with no solid elements.}.
	Now there are two cases: if $\size(g(y)) \ge p$, then we move $p$ to $g(y)$ and move liquid elements of the same total size back to $g(x)$. This preserves $\cost(g)$ (think that these elements went first in the linear orders on both $g(x)$ and $g(y)$). On the other hand, if $\size(g(y)) < p$, then we move $p$ to $g(y)$ and move all liquid elements from $g(y)$ to $g(x)$. This even decreases $\cost(g)$.\footnote{Formally, select $\tau > 0$ so that $g$ is constant on $[x,x+\tau)$ and on $[y,y+\tau)$; for $x' \in [x,x+\tau)$, replace $g(x')$ with $g(x') \setminus \{p\} \cup g(y)$, and for $y' \in [y,y+\tau)$, replace $g(y')$ with $\{p\}$. The cost of $g$ then decreases by $\tau (\size(g(x)) - p) (p - \size(g(y))) > 0$.}
	
	At this point, $f$ and $g$ have the same solid elements, each of which appears as the only solid element in patterns containing it in both $f$ and $g$. Thus we can now sort $g$ so that for each solid element $p > 0$, it appears in the same positions in $f$ and in $g$. This operation preserves $\cost(g)$, and thus the entire third step does not decrease the ratio of costs.
	\end{enumerate}
\end{proof}

\subsection{The Final Form}

In our last transformation, we will guarantee that in $g$, each large element is the only member of a~pattern which contains it.
(Intuitively, we do this because such CFPs are the ones which maximize the ratio of costs.)
Namely, we prove Lemma\nobreakspace \ref {lem:introduce_t}, a strengthened version of~Lemma\nobreakspace \ref {lem:introduce_m};
note that condition \textit{(c')} below is stronger than condition \textit{(c)} of~Lemma\nobreakspace \ref {lem:introduce_m} (there, $g'$ could have had patterns with both liquid and solid elements), and that condition \textit{(d)} is new.
Figure\nobreakspace \ref {fig:idealdistro} shows the difference between CFPs postulated by Lemmas\nobreakspace \ref {lem:introduce_m} and\nobreakspace  \ref {lem:introduce_t}.

\newcommand{\introducetstatement}{
	Let $(f,g)$ be as postulated in~Lemma\nobreakspace \ref {lem:introduce_m}. Then there exists another CFP $(f',g')$ such that:
	\begin{enumerate}[(a)] 
		\item $f_r'$ is constant.
		\item $f_1' = g_1'$.
		\item[(c')] There exists $t \in [0,1]$ such that:
			\begin{itemize}
				\item for $x \in [0,t)$, $f'(x)$ has liquid elements and exactly one solid element,
				\item for $x \in [t,1)$, $f'(x)$ has only liquid elements,
				\item for $x \in [0,t)$, $g'(x)$ has exactly one solid element (and no liquid ones),
				\item for $x \in [t,1)$, $g'(x)$ has only liquid elements.
			\end{itemize}
		\item[(d)] The function $\size_{g'}$ is constant on $[t,1)$ (i.e., the liquid part of $g$ is uniform).
	\end{enumerate}
	Moreover, \[ \frac{\cost(f')}{\cost(g')} \ge \min \left( 2, \frac{\cost(f)}{\cost(g)} \right). \]
}
\begin{lemma} \label{lem:introduce_t}
	\introducetstatement
\end{lemma}
\input{idealdistro}
\begin{proof}
	Let us begin by giving a proof outline;
	see also~Figure\nobreakspace \ref {fig:idealdistro} for an illustration.
	
\begin{itemize}
	\item Our primary objective is to make $g_1$ equal to $\size_g$ on $[0,t)$ (where $t$ is to be determined). To this end, we increase the sizes of the solid elements in $g(x)$ for $x \in [0,t)$ and at the same time decrease the total sizes of liquid elements for these $g(x)$ (which keeps $\size_g$ unchanged). To offset this change, we decrease the sizes of solid elements in $g(y)$ for $y \in [t,m)$ (and also increase the total sizes of liquid elements there). We also modify the sizes of solid elements in $f$ so as to keep $f$ and $g$ compatible and preserve the properties \textit{(a)-(b)}.
	\item The threshold $t$ is defined so that after we finish this process, the solid elements in $g(x)$ for $x \in [0,t)$ will have filled
out the entire patterns $g(x)$, and the solid elements in $g(y)$ for $y \in [t,m)$ will have disappeared.
	\item Our main technical claim is that this process does not invalidate the ratio of costs.
	\item Finally, we can easily ensure condition \textit{(d)} by levelling $g$ on $[t,1)$, which only decreases its cost.
\end{itemize}
Now we proceed to the details.
First we define the threshold $t \in [0,m]$ as a solution to the equation
\[ \int_0^t g_r(x) \, dx = \int_t^m g_1(x) \, dx, \]
which exists because the functions under both integrals are nonnegative. The left-hand side will be the total increase in sizes of solid elements in $g(x)$ for $x \in [0,t)$ and the right-hand side will be the total decrease in sizes of solid elements in $g(y)$ for $y \in [t,m)$ (these elements will disappear completely).

Now we will carry out the process that we have announced in the outline. Namely, while there exists $x \in [0,t)$ with $g_r(x) > 0$, do the following:
\begin{itemize}
	\item find $y \in [t,m)$ with $g_1(y) > \eps$ (i.e., $g(y)$ where the solid element has not been eradicated yet),
	\item increase the size of the solid element in $g(x)$ by $\eps$,
	\item do the same in $f$,
	\item decrease the size of the solid element in $g(y)$ by $\eps$,
	\item do the same in $f$,
	\item move one liquid element (of size $\eps$) from $g(x)$ to $g(y)$.
\end{itemize}

Formally, as usual, we find $\tau > 0$ such that $f$ and $g$ are constant on $[x,x+\tau)$ and on $[y,y+\tau)$ and we do this in all of these patterns.

Note that the following invariants are satisfied after each iteration:
\begin{itemize}
	\item $f_1 = g_1$,
	\item $f_r$ does not change,
	\item $\size_g$ does not change,
	\item $\int_0^m g_1(x) \, dx$ does not change,
	\item $g_1$ can only increase on $[0,t)$ and it can only decrease on $[t,m)$,
	\item $f_1(x) \ge f_1(y)$ for all $x \in [0,t)$ and $y \in [t,m)$.\footnote{This is because $f_1 = g_1$ was initially non-increasing and since then it has increased on $[0,t)$ and decreased on $[t,m)$.}
\end{itemize}

By the definition of $t$, when this process ends, the patterns $g(x)$ for $x \in [0,t)$ contain only a single solid element, while $g(x)$ for $x \in [t,m)$ (thus also for $x \in [t,1)$) contain no solid elements. Since $f_1 = g_1$, the patterns $f(x)$ also have only liquid elements for $x \in [t,1)$. Thus properties \textit{(a), (b)} and \textit{(c')} are satisfied. We reason about the ratio of costs in the following two technical claims:

\begin{claim} \label{fact:cost_two_alpha}
	In a single iteration, $\cost(f)$ increases by $2 \alpha$ and $\cost(g)$ increases by $\alpha$, for some $\alpha \ge 0$.
\end{claim}
\begin{proof}
	The patterns have changed so that:
	\begin{itemize}
		\item $f(x)$ had $f_1(x)$ increased by $\eps$,
		\item $f(y)$ had $f_1(y)$ decreased by $\eps$,
		\item $g(x)$ had $g_1(x)$ increased by $\eps$ and one liquid element removed,
		\item $g(y)$ had $g_1(y)$ decreased by $\eps$ and one liquid element added.
	\end{itemize}
	\ifsoda
\else
		Since the order of elements does not matter, in computing $\cost(f)$ we think that the solid element goes last in the linear order:
	
		\begin{equation*} \begin{split}
			\Delta \cost(f(x)) &= (f_1(x) + \eps) (\size(f(x)) + \eps) - f_1(x) \size(f(x)) = \eps \rb{ \size(f(x)) + f_1(x) + \eps}, \\
			\Delta \cost(f(y)) &= (f_1(y) - \eps) (\size(f(y)) - \eps) - f_1(y) \size(f(y)) = \eps \rb{ - \size(f(y)) - f_1(y) + \eps},
		\end{split} \end{equation*}
		\begin{equation*} \begin{split}
			\Delta \cost(f(x)) + \Delta \cost(f(y)) &= \eps \rb{ \size(f(x)) - \size(f(y)) + f_1(x) - f_1(y) + 2 \eps } \\
			&= 2 \eps \rb{ f_1(x) - f_1(y) + \eps },
		\end{split} \end{equation*}
	\fi
	where in the last line we used that $\size(f(x)) - \size(f(y)) = f_1(x) + f_r(x) - f_1(y) - f_r(y) = f_1(x) - f_1(y)$ since $f_r$ is constant.
	
	In computing $\cost(g)$, we think that the solid element and the one liquid element which was added or removed go first (and other elements are unaffected since $\size_g$ is preserved):
	\ifsoda
\else
		\begin{align*}
			\Delta \cost(g(x)) &= (g_1(x) + \eps)^2 - (g_1(x)^2 + \eps (g_1(x) + \eps)) = \eps g_1(x), \\
			\Delta \cost(g(y)) &= ((g_1(y) - \eps)^2 + \eps g_1(y)) - g_1(y)^2.
		\end{align*}
	\fi
	Adding up we have:
	\ifsoda
\else
		\begin{align*}
			\Delta \cost(g(x)) + \Delta \cost(g(y)) = \eps \rb{ g_1(x) - g_1(y) + \eps } = \eps \rb{ f_1(x) - f_1(y) + \eps },
		\end{align*}
	\fi
	where  we used that $f_1 = g_1$. Thus we have that
	\ifsoda
\else
		\[ \Delta \cost(f(x)) + \Delta \cost(f(y)) = 2 \sb {\Delta \cost(g(x)) + \Delta \cost(g(y))} \]
	\fi
 
	and we get the statement by setting
	\ifsoda
\else
		\[ \alpha = \tau \rb{ \Delta \cost(g(x)) + \Delta \cost(g(y))} = \tau \eps \rb{ f_1(x) - f_1(y) + \eps } \ge 0 \]
	\fi
	(recall that $\tau$ is the fraction of patterns where we increase $g_1$; nonnegativity follows by the last invariant above).
\end{proof}

\begin{claim}
	Let $(f',g')$ be the CFP obtained at this point and $(f,g)$ be the original CFP. Then
	\[ \frac{\cost(f')}{\cost(g')} \ge \min \left( 2, \frac{\cost(f)}{\cost(g)} \right). \]	
\end{claim}
\begin{proof}
	By~\protect \MakeUppercase {C}laim\nobreakspace \ref {fact:cost_two_alpha}, we have $\cost(f') = \cost(f) + 2 \beta$ and $\cost(g') = \cost(g) + \beta$ for some $\beta \ge 0$ (which is the sum of $\alpha$'s from~\protect \MakeUppercase {C}laim\nobreakspace \ref {fact:cost_two_alpha}). Now there are two cases:
	\begin{itemize}
		\item if $\frac{\cost(f)}{\cost(g)} \le 2$, then $\frac{\cost(f) + 2 \beta}{\cost(g) + \beta} \ge \frac{\cost(f)}{\cost(g)}$,
		\item if $\frac{\cost(f)}{\cost(g)} \ge 2$, then $\frac{\cost(f) + 2 \beta}{\cost(g) + \beta} \ge 2$ (even though the ratio decreases, it stays above $2$).
	\end{itemize}
\end{proof}

Finally, as the last step, we equalize the total sizes of liquid elements in $g(x)$ for $x \in [t,1)$ (by moving liquid elements from larger patterns to smaller patterns, until all are equal), thus satisfying property \textit{(d)}. Clearly, this can only decrease the cost of $g$ (by minimizing the variance of $\size_g$ on the interval $[t,1)$), so the ratio increases and Lemma\nobreakspace \ref {lem:introduce_t} follows.
\end{proof}

Note that Lemma\nobreakspace \ref {lem:introduce_t} does not guarantee that the ratio of costs increases; we only claim that it either increases, or it is now more than $2$. However, we will shortly show in~Lemma\nobreakspace \ref {lem:bound_ratio} that the ratio is actually much below $2$, so the latter is in fact impossible.

Now that we have our ideal CFP, we can finally bound its cost ratio.

\newcommand{\boundratiostatement}{
	Given a CFP $(f,g)$ as postulated in~Lemma\nobreakspace \ref {lem:introduce_t} (see the lower part of~Figure\nobreakspace \ref {fig:idealdistro}), we have \[ \frac{\cost(f)}{\cost(g)} \le \frac{1 + \sqrt{2}}{2}. \]
}
\begin{lemma} \label{lem:bound_ratio}
	\boundratiostatement
\end{lemma}
The proof proceeds in two simple steps: first, we argue that we can assume without loss of generality that there is only a single large element (i.e., $f_1 = g_1$ is constant on $[0,t)$). Then, for such pairs of functions, the ratio is simply a real function of three variables whose maximum is easy to compute.
\begin{proof}
As a first step, we assume without loss of generality that there is only a single large element (i.e., $f_1 = g_1$ is constant on $[0,t)$). This is due to the fact that both $f$ and $g$ can be written as a weighted average of functions with a single large element. Formally, let $\ell_1, \ell_2, ...$ be the step lengths of $f_1$ on $[0,t)$, so that $\sum_i \ell_i = t$ and $f_1$ is constant on $[0,\ell_1)$, on $[\ell_1,\ell_1+\ell_2)$ and so on. Define $f^i$ to be $f$ with the whole $f_1$ on $[0,t)$ replaced by the $i$-th step of $f_1$, i.e.,
\[ f^i(x) = \begin{cases} f(\ell_1 + ... + \ell_{i-1}) & \text{for } x \in [0,t), \\ f(x) & \text{for } x \in [t,1). \end{cases} \]
Define $g^i$ similarly.
\ifsoda
\else
	Then
	\begin{align*}
		\cost(f) &= \sum_i \ell_i \cost(f(\ell_1 + ... + \ell_{i-1})) + (1-t) \cdot \cost(f(1)) \\
		         &= \sum_i \frac{\ell_i}{t} \sb{ t \cdot \cost(f(\ell_1 + ... + \ell_{i-1})) + (1-t) \cdot \cost(f(1)) } \\
		         &= \sum_i \frac{\ell_i}{t} \cost(f^i)
	\end{align*}
\fi
and
similarly
$\cost(g) = \sum_i \frac{\ell_i}{t} \cost(g^i)$.
Thus, if we have $\frac{\cost(f^i)}{\cost(g^i)} \le \frac{1 + \sqrt{2}}{2}$ for each $i$, then
\ifsoda
\else
	\[ \frac{\cost(f)}{\cost(g)} = \frac{\sum_i \frac{\ell_i}{t} \cost(f^i)}{\sum_i \frac{\ell_i}{t} \cost(g^i)} \le \frac{\sum_i \frac{\ell_i}{t} \frac{1 + \sqrt{2}}{2} \cost(g^i)}{\sum_i \frac{\ell_i}{t} \cost(g^i)} = \frac{1+\sqrt{2}}{2}. \]
\fi
So we assume that $f_1$ is constant on $[0,t)$ (i.e., the shaded areas in~Figure\nobreakspace \ref {fig:idealdistro} are rectangles). Let $\gamma = f_1(0)$ be the large element and $\lambda$ be the total mass of liquid elements (the same in $f$ as in $g$), i.e., $\lambda = f(1) = (1-t) g(1)$. In the limit $\eps \to 0$ we have
\ifsoda
\else
	\[ \frac{\cost(f)}{\cost(g)} = \frac{ t \rb{ \gamma^2 + \int_0^\lambda \rb{ \gamma + x } dx } + (1-t) \int_0^\lambda x dx }{ t \gamma^2 + (1-t) \int_0^{g(1)} x \, dx } = \frac{ t \rb{ \gamma^2 + \gamma \lambda + \frac{\lambda^2}{2} } + (1-t) \frac{\lambda^2}{2} } {t \gamma^2 + (1-t) \frac{\rb{\frac{\lambda}{1-t}}^2}{2}} = \frac{t \gamma^2 + t \gamma \lambda + \frac{\lambda^2}{2}}{t \gamma^2 + \frac{\lambda^2}{2(1-t)}}
\]
\fi
and we need to prove that this expression is at most $\frac{1+\sqrt{2}}{2}$ for all $t \in [0,1)$, $\gamma \ge 0$ and $\lambda \ge 0$. So we
want to show
\[ t \gamma^2 + t \gamma \lambda + \frac{\lambda^2}{2} \le \frac{1+\sqrt{2}}{2} \rb{ t \gamma^2 + \frac{\lambda^2}{2(1-t)} }, \]
that is,
\[ \lambda^2 \rb{ \frac{1+\sqrt{2}}{4(1-t)} - \frac 12 } - \lambda \cdot t \gamma + \frac{\sqrt{2}-1}{2} t \gamma^2 \ge 0. \]
Note that $\frac{1+\sqrt{2}}{4(1-t)} - \frac 12 > 0$ for $t \in [0,1)$, so this is a quadratic polynomial in $\lambda$ whose minimum value (over $\lambda \in \bR$) is
\ifsoda
\else
	\[ \frac{\sqrt{2} - 1}{2} t \gamma^2 - \frac{t^2 \gamma^2}{4 \rb{ \frac{1+\sqrt{2}}{4(1-t)} - \frac 12 }} = t \gamma^2 \rb{ \frac{\sqrt{2} - 1}{2} - \frac{t}{ \frac{1+\sqrt{2}}{1-t} - 2 } } \]
\fi
and we should prove that this is nonnegative. If $t = 0$ or $\gamma = 0$, then this is clearly true; otherwise we multiply both sides of the inequality by $\frac{1-t}{t \gamma^2} \rb{ \frac{1+\sqrt{2}}{1-t} - 2 }$ (a positive number) and after some calculations we are left with showing
\[ t^2 + \rb{\sqrt{2} - 2} t + \frac{3 - 2 \sqrt{2}}{2} \ge 0 \]
but this is again a quadratic polynomial, whose minimum is $0$.
\end{proof}

\begin{proof}[Proof of Theorem\nobreakspace \ref {thm:maintech}]
Now it is straightforward to see that
\ifsoda
\else
	\[ \frac{1+\sqrt{2}}{2} \ge \frac{\cost(f')}{\cost(g')} \ge \min \left( 2, \frac{\cost(f)}{\cost(g)} \right) = \min \left( 2, \frac{\sum_i \yout_i \cost(C_i)}{\sum_i \yin_i \cost(C_i)} \right) \]
\fi
where $(f,g)$ is produced from $(\yout,\yin)$ as in~\protect \MakeUppercase {F}act\nobreakspace \ref {fact:output_is_nice} and $(f',g')$ is produced from $(f,g)$ by applying Lemmas\nobreakspace \ref {lem:worst_case_output},  \ref {lem:introduce_m} and\nobreakspace  \ref{lem:introduce_t}; the first inequality is by~Lemma\nobreakspace \ref {lem:bound_ratio}. It follows that either $\frac{1+\sqrt{2}}{2} \ge 2$ (false) or $\frac{1+\sqrt{2}}{2} \ge \frac{\sum_i \yout_i \cost(C_i)}{\sum_i \yin_i \cost(C_i)}$.
\end{proof}

On the other hand, to conclude the proof of Theorem\nobreakspace \ref {thm:main}, we have the following lemma regarding the tightness of our analysis of Algorithm~\ref{alg:round}:
\newcommand{\tightstatement}{
	For any $\delta>0$, there is an instance $I$ of $R | | \sum p_j C_j$ whose optimal value is $c$, 
	and an optimal Configuration-LP solution whose objective value is also $c$, such that the rounded solution returned by
Algorithm
	\ref{alg:round} has
	cost at least $(\frac{1 + \sqrt{2}}{2}-\delta)c$.
}
\begin{lemma}\label{lem:tight}
	\tightstatement
\end{lemma}
\begin{proof}
The intuitive explanation is that the bound in~Lemma\nobreakspace \ref {lem:bound_ratio} is tight and thus there exists a CFP in the final form (as postulated by Lemma\nobreakspace \ref {lem:introduce_t}) with ratio exactly $\frac{1+\sqrt{2}}{2}$. Furthermore, this CFP indeed almost corresponds to an instance of $R | | \sum p_j C_j$ (except for the fact that the parameters which maximize the ratio are irrational). We make this intuition formal in the following.

Let 
\[h(t,\gamma,\lambda)=\frac{t \gamma^2 + t \gamma \lambda + \frac{\lambda^2}{2}}{t \gamma^2 + \frac{\lambda^2}{2(1-t)}}\]
be the function specifying the ratio of CFPs in the final form.
In Lemma\nobreakspace \ref {lem:bound_ratio}, we proved that $h(t,\gamma,\lambda)\leq \frac{1}{2}+\frac{1}{\sqrt{2}}$ for all $t,\gamma,\lambda\in[0,1]$. To
begin, let us fix the following maximizer $(t^\star,\lambda^\star,\gamma^\star)$ of $h$: $t^\star=1-\frac{1}{\sqrt{2}}$,
$\gamma^\star=\frac{1}{2}$ and $\lambda^\star=\frac{\sqrt{2}-1}{2}$; we have
$h(t^\star,\lambda^\star,\gamma^\star)=\frac{1}{\sqrt{2}}+\frac{1}{2}$.

Let us choose a small rational $\eta$. Next, let us fix rationals $\tilde{t}\in[t^\star-\eta,t^\star]$,
$\tilde{\lambda}\in[\lambda^\star-\eta,\lambda^\star]$ and $\tilde{\gamma}\in[\gamma^\star,\gamma^\star+\eta]$
such that $h(\tilde{t},\tilde{\gamma},\tilde{\lambda})\ge \frac{1+\sqrt{2}}{2}-\eta$.
Then, there exist positive integers $k$, $T$ and $\Lambda$ such that $T=\tilde{t}k$ and $\Lambda=k\tilde{\lambda}$.
Finally, select a small rational $\epsilon\le \eta$ such that
$\eps=\frac{\tilde{\lambda}}{k_1(1-\tilde{t})}$, for some integer $k_1>0$, and 
 $\eps=\frac{\tilde{\lambda}}{k_2}$, for some integer $k_2>0$. 
 
 Next, we create an instance 
$I_\eps$ of $R | | \sum p_j C_j$ which consists of $k$ machines, a set $\mathcal{T}$ of $T$ jobs of size $\tilde{\gamma}$ each, and a set
$\mathcal{L}$ of $\Lambda/\eps$ jobs of size $\eps$ each; any job can be assigned to any machine. 
An optimal solution to this instance will assign the jobs from $\mathcal{T}$ alone on $T$ machines, and distribute the jobs from
$\mathcal{L}$ evenly on the rest of the machines (i.e., these machines will all receive $\frac{\tilde{\lambda}}{\eps}$ jobs of size
$\eps$ each). The fact that this is an optimal solution follows in a straightforward manner from the following two observations:
\begin{itemize}
 \item A solution which assigns a job from $\mathcal{L}$ on the same machine as a job from $\mathcal{T}$ is sub-optimal: indeed, the
average makespan is less than $\tilde{\gamma}$ (in fact, it is exactly $\tilde{t}\tilde{\gamma}+\tilde{\lambda}$, which is at most
$\tilde{\gamma}$, due to the fact that $t^\star\gamma^\star+\lambda^\star < \gamma^\star$ and due to the intervals which we choose
$\tilde{t}$, $\tilde{\gamma}$ and $\tilde{\lambda}$ from), which implies that we can always reassign such a job to a machine with smaller
makespan, thus decreasing the solution cost.
\item Similarly, in any optimal solution, jobs from $\mathcal{T}$ are not assigned on the same machine.
\end{itemize}

Now, consider the Configuration-LP solution $y_\eps$ which assigns to every machine a configuration which consists of a single job from
$\mathcal{T}$ (i.e., of cost $\tilde{\gamma}^2$) with probability $\tilde{t}$, and a configuration which consists of
$\frac{\tilde{\lambda}}{(1-\tilde{t})\eps}$ jobs  from $\mathcal{L}$ 
(i.e., of cost 
$\sum_{1\leq i\leq\frac{\tilde{\lambda}}{(1-\tilde{t})\eps}}\sum_{1\leq j<i}\eps^2=
\frac{\tilde{\lambda}^2}{2(1-\tilde{t})^2}+\frac{\tilde{\lambda}}{2(1-\tilde{t})}\eps$)
with probability $1-\tilde{t}$; clearly, the
cost of this LP solution is equal to that of any optimal integral solution 
(in fact, the LP solution is a convex combination of all integral optimal solutions). Furthermore, this LP solution is optimal (one
can see this by applying the reasoning we used for the integral optimum to all the configurations in the support of a fractional solution).

Algorithm \ref{alg:round} will assign to any machine a configuration which consists of a single job from
$\mathcal{T}$ and $\tilde{\lambda}/\eps$ jobs from $\mathcal{L}$ 
(i.e., of cost 
$\tilde{\gamma}(\tilde{\gamma}+\tilde{\lambda})+
\sum_{1\leq i\leq \frac{\tilde{\lambda}}{\eps}}\sum_{1\leq j<i} \eps^2= 
\tilde{\gamma}(\tilde{\gamma}+\tilde{\lambda})+\frac{\tilde{\lambda}^2}{2}+\frac{\tilde{\lambda}}{2}\eps$) 
with probability $\tilde{t}$, and a configuration which consists of
$\tilde{\lambda}/\eps$ jobs from $\mathcal{L}$ 
(i.e., of cost $\sum_{1\leq i\leq \frac{\tilde{\lambda}}{\eps}}\sum_{1\leq j<i}
\eps^2=\frac{\tilde{\lambda}^2}{2}+\frac{\tilde{\lambda}}{2}\eps$)
with probability $1-\tilde{t}$. To see this, first observe that every machine has a total fractional assignment of jobs from
$\mathcal{T}$ equal to $\tilde{t}$, and a total fractional assignment of jobs from $\mathcal{L}$ equal to $\frac{\tilde{\lambda}}{\eps}$.
Therefore, the first bucket created by Algorithm \ref{alg:round} for any machine will contain a $\tilde{t}$-fraction of $\mathcal{T}$-jobs
and an $(1-\tilde{t})$-fraction of $\mathcal{L}$-jobs, and the rest of the buckets will be filled up with $\mathcal{L}$-jobs (since
$\frac{\tilde{\lambda}}{\eps}$ is an integer, the last bucket will be filled up to a $\tilde{t}$-fraction). This implies that, in a
worst-case output distribution, with probability $\tilde{t}$ any machine receives a $\mathcal{T}$-job and $\mathcal{L}$-jobs of total size
$\tilde{\lambda}$, and with probability $(1-\tilde{t})$ it receives $\mathcal{L}$-jobs of total size $\tilde{\lambda}$.

Now, the ratio of the expected cost of the returned solution
to the LP cost, for any machine, is then
\ifsoda
\else
	\begin{align*}
	\frac{\tilde{t} \tilde{\gamma}^2 + \tilde{t} \tilde{\gamma} \tilde{\lambda}+ \frac{\tilde{\lambda}^2}{2}+
	\frac{\tilde{\lambda}}{2}\eps}{\tilde{t} \tilde{\gamma}^2 +
	\frac{\tilde{\lambda}^2}{2(1-\tilde{t})}+\frac{\tilde{\lambda}}{2}\eps}
	\ge
	\frac{\tilde{t} \tilde{\gamma}^2 + \tilde{t} \tilde{\gamma}\tilde{\lambda}+ \frac{\tilde{\lambda}^2}{2}}{
	\tilde{t} \tilde{\gamma}^2 + \frac{\tilde{\lambda}^2}{2(1-\tilde{t})}+\tilde{\lambda}\frac{\eps}{2}}
	=
	h(\tilde{t},\tilde{\gamma},\tilde{\lambda})\frac{\tilde{t} \tilde{\gamma}^2 +
	\frac{\tilde{\lambda}^2}{2(1-\tilde{t})}}{\tilde{t} \tilde{\gamma}^2 +
	\frac{\tilde{\lambda}^2}{2(1-\tilde{t})}+\frac{\tilde{\lambda}}{2}\eps}
	\end{align*}
\fi
which is at least $\frac{1+\sqrt{2}}{2}-\delta$ if we pick $\eps$ and $\eta$ small enough; since the cost of the LP solution is equal
to
that of any optimal integral solution, the claim follows.
\end{proof}

It is interesting to note that, given the instance and LP solution from the proof of Lemma~\ref{lem:tight}, any random assignment produced by Algorithm~\ref{alg:round} will
assign the same amount of small jobs to all the machines, while it will assign a large job to a $\tilde{t}$-fraction of the machines.
Therefore derandomizing Algorithm \ref{alg:round} by picking the {\em best possible} matching (instead of picking one at random) will
not improve upon the $\frac{1+\sqrt{2}}{2}$ ratio.

Theorem\nobreakspace \ref {thm:maintech} and\nobreakspace Lemma\nobreakspace \ref {lem:tight} together imply Theorem\nobreakspace \ref{thm:main}.

%% file: define_m.tex
\begin{figure*}[t]
\begin{minipage}[t]{\linewidth}
    \begin{tikzpicture}[scale=0.3]


\filldraw[draw=none,fill=gray] (0,0) -- (0,9) -- (20,4) -- (20,0);

\pattern[pattern=north east lines,draw=none] (15,0) -- (15,2.25) -- (20,1) -- (20,0); 
\filldraw[draw=none,fill=gray] (20,1) -- (15,2.25) -- (15,3.25) -- (20,3.25); 
\filldraw[draw=none, fill=lightgray] (0,0) -- (0,6) -- (15,3) -- (15,0);


\draw (-1,0) -- (0,0) node[anchor=north,yshift=-2] {$0$} -- (15,0) node[anchor=north,yshift=-2] {$m$} -- (21,0) node[anchor=north,yshift=-2]
{$1$};

\draw (0,0) -- (0,6) -- (5,5) node[anchor=south, yshift=-2] {\tiny{$f_1(x)$}} -- (20,2) -- (20,0); 
\draw[dashed] (0,9) -- (5,8)node[anchor=south, yshift=1, xshift=2] {\tiny{$f_1(x)+f_r(0)$}} -- (15,6); 

\draw (0,6) -- (0,9)-- (5,7.75) node[anchor=north, yshift=-2, xshift=5] {\tiny{$\size(f(x))$}} -- (20,4) -- (20,0); 
\draw[dashed] (0,6) -- (5,4.75) node[anchor=north, yshift=-2, xshift=-8] {\tiny{$\size(f(x))-f_r(0)$}} --(20,1); 

\draw[dashed] (15,0) -- (15,6);

\draw [decorate,decoration={brace,amplitude=5},xshift=0,yshift=0pt]
(0,6) -- (0,9)node [black,midway,xshift=-15] {\tiny $f_r(0)$};

\node at (25,5) {\Huge$\Rightarrow$};
\begin{scope}[shift={(30,0)}]
 


\filldraw[draw=none,fill=gray] (0,0) -- (0,9) -- (15,6) -- (15,5.25) -- (20,4) -- (20,1) -- (15,2.25) -- (15,0);

\pattern[pattern=north east lines,draw=none] (0,9) -- (15,6) -- (15,5.25); 
\filldraw[draw=none,fill=white] (0,9) -- (15,6) -- (15,10) -- (0,10); 
\filldraw[draw=none,fill=gray] (0,9) -- (15,5.25) -- (0,5.25);
\filldraw[draw=none, fill=lightgray] (0,0) -- (0,6) -- (15,3) -- (15,0);

\draw (-1,0) -- (0,0) node[anchor=north,yshift=-2] {$0$} -- (15,0) node[anchor=north,yshift=-2] {$m$} -- (21,0) node[anchor=north,yshift=-2]
{$1$};

\draw[dashed] (0,9) -- (15,5.25);
\draw (0,0) -- (0,9) -- (15,6) -- (15,5.25) -- (20,4);
\draw (20,0) -- (20,4);
\draw (0,0) -- (0,6) -- (5,5)-- (20,2) -- (20,0); 

\draw[dashed] (0,6) -- (15,2.25);
\draw (15,2.25) -- (20,1); 

\draw (15,0) -- (15,2.25);
\draw[dashed] (15,2.25) -- (15,5.25);


\end{scope}

\end{tikzpicture}
    \end{minipage}%
          \caption{The main step in the proof of~Lemma\nobreakspace \ref {lem:introduce_m}. 
          The left picture shows $f$ after the liquification of all elements except one largest element ($f_1$) for $x \in [0,m)$. The right picture shows $f$ after the movement of liquid elements between the striped regions. In both pictures, the light-gray areas contain single large elements, and the dark-gray areas are liquid elements. Note that there are two pairs of parallel lines in the pictures; the vertical distance between each pair is $f_r(0)$. The threshold $m$ is defined so that the striped regions have equal areas (thus $f_r'$ is made constant by the movement of liquid elements) and so that moving the liquid elements can only increase the cost. The height of the upper striped area is $f_1(x) + f_r(0) - \size(f(x)) = f_r(0) - f_r(x)$ for $x \in [0,m)$ and the height of the lower striped area is $\size(f(x)) - f_r(0)$ for $x \in [m,1)$.
	All functions are stepwise-constant, but drawn here using straight downward lines for simplicity; also, the rightmost dark-gray part will ``fall down'' (forming a $(1-m) \times f_r(0)$ rectangle).
}
    
    \label{fig:define_m}

\end{figure*}
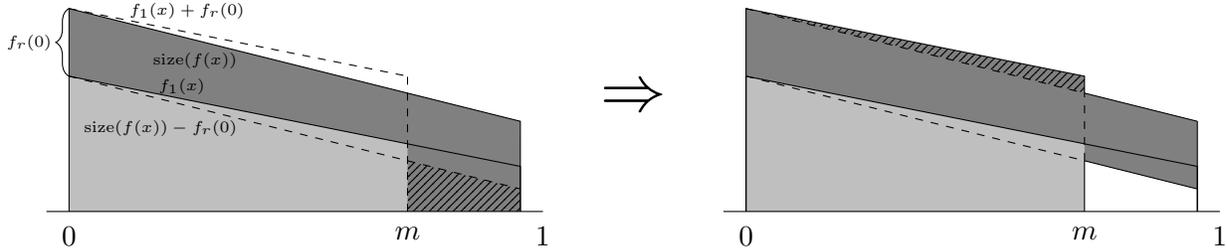

%% file: idealdistro.tex
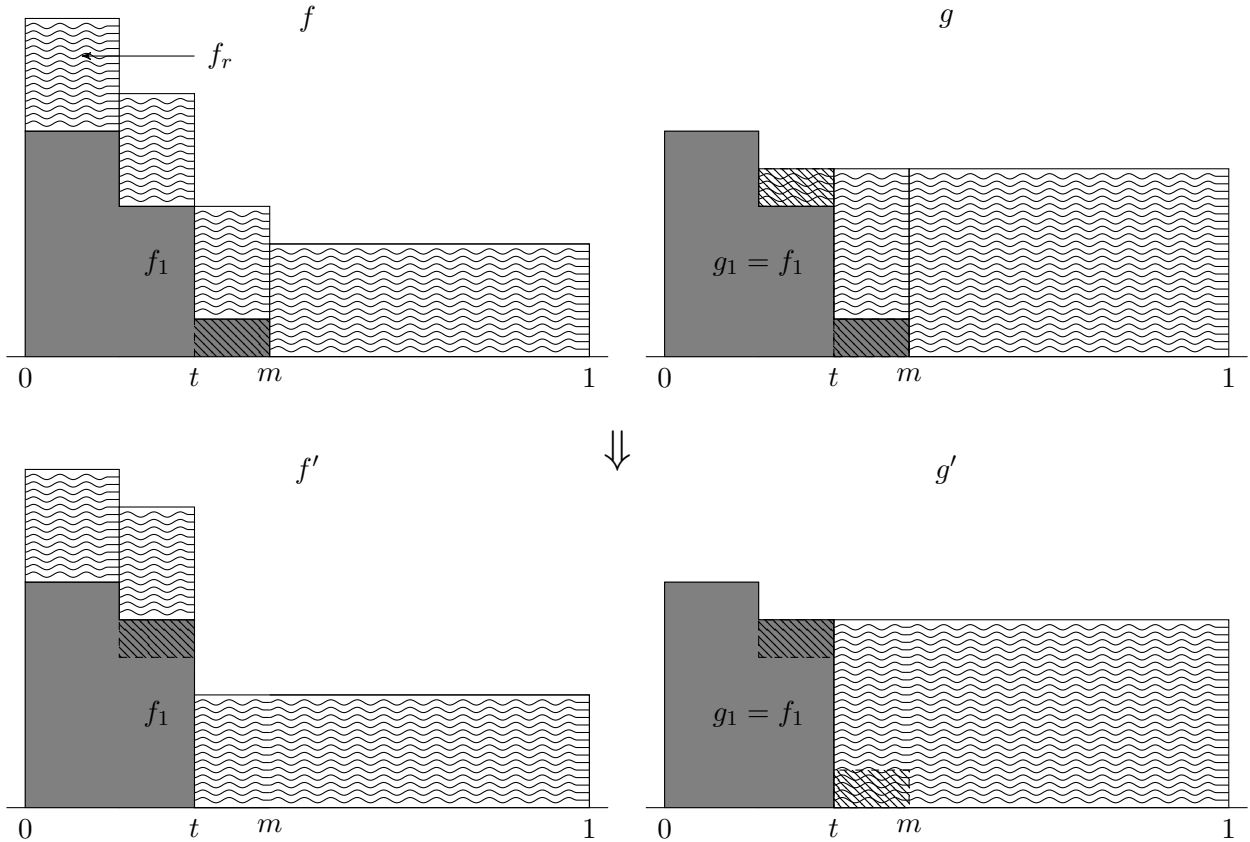
\begin{figure*}[t!]
\begin{minipage}[t]{\linewidth}
    \begin{tikzpicture}[scale=0.5]
\node at (18.25,-1.5) {\LARGE$\Downarrow$};

\begin{scope}
  \filldraw[fill=gray,draw=gray] (2.5,1) rectangle  (5,7);
 \filldraw[fill=gray,draw=gray]  (5,1) rectangle (7,5);
 \filldraw[fill=gray,draw=gray]  (7,2) rectangle (9,1);
 \filldraw[draw=none,fill=none,pattern=north west lines,dashed] (7,2) rectangle (9,1);

\filldraw[fill=white,draw=black] (2.5,7) rectangle (5,10);
\draw[snake=snake,segment amplitude=1] (2.5,7.2) -- (5,7.2);\draw[snake=snake,segment amplitude=1] (2.5,8.2) -- (5,8.2);
\draw[snake=snake,segment amplitude=1] (2.5,7.4) -- (5,7.4);\draw[snake=snake,segment amplitude=1] (2.5,8.4) -- (5,8.4);
\draw[snake=snake,segment amplitude=1] (2.5,7.6) -- (5,7.6);\draw[snake=snake,segment amplitude=1] (2.5,8.6) -- (5,8.6);
\draw[snake=snake,segment amplitude=1] (2.5,7.8) -- (5,7.8);\draw[snake=snake,segment amplitude=1] (2.5,8.8) -- (5,8.8);
\draw[snake=snake,segment amplitude=1] (2.5,8) -- (5,8);\draw[snake=snake,segment amplitude=1] (2.5,9.2) -- (5,9.2);
\draw[snake=snake,segment amplitude=1] (2.5,9.4) -- (5,9.4);\draw[snake=snake,segment amplitude=1] (2.5,9.6) -- (5,9.6);
\draw[snake=snake,segment amplitude=1] (2.5,9.8) -- (5,9.8);\draw[snake=snake,segment amplitude=1] (2.5,9) -- (5,9);

 \filldraw[white, draw=black] (5,5) rectangle (7,8);
\draw[snake=snake,segment amplitude=1] (5,5.2) -- (7,5.2);\draw[snake=snake,segment amplitude=1] (5,6.2) -- (7,6.2);
\draw[snake=snake,segment amplitude=1] (5,5.4) -- (7,5.4);\draw[snake=snake,segment amplitude=1] (5,6.4) -- (7,6.4);
\draw[snake=snake,segment amplitude=1] (5,5.6) -- (7,5.6);\draw[snake=snake,segment amplitude=1] (5,6.6) -- (7,6.6);
\draw[snake=snake,segment amplitude=1] (5,5.8) -- (7,5.8);\draw[snake=snake,segment amplitude=1] (5,6.8) -- (7,6.8);
\draw[snake=snake,segment amplitude=1] (5,6) -- (7,6);
\draw[snake=snake,segment amplitude=1] (5,7) -- (7,7);\draw[snake=snake,segment amplitude=1] (5,7.2) -- (7,7.2);
\draw[snake=snake,segment amplitude=1] (5,7.4) -- (7,7.4);\draw[snake=snake,segment amplitude=1] (5,7.6) -- (7,7.6);
\draw[snake=snake,segment amplitude=1] (5,7.8) -- (7,7.8);

\begin{scope}[shift={(0,-2)}]
 \filldraw[white, draw=black] (7,7) rectangle (9,4);
\draw[snake=snake,segment amplitude=1] (7,4.2) -- (9,4.2);\draw[snake=snake,segment amplitude=1] (7,5.2) -- (9,5.2);
\draw[snake=snake,segment amplitude=1] (7,4.4) -- (9,4.4);\draw[snake=snake,segment amplitude=1] (7,5.4) -- (9,5.4);
\draw[snake=snake,segment amplitude=1] (7,4.6) -- (9,4.6);\draw[snake=snake,segment amplitude=1] (7,5.6) -- (9,5.6);
\draw[snake=snake,segment amplitude=1] (7,4.8) -- (9,4.8);\draw[snake=snake,segment amplitude=1] (7,5.8) -- (9,5.8);
\draw[snake=snake,segment amplitude=1] (7,5) -- (9,5);
\draw[snake=snake,segment amplitude=1] (7,6) -- (9,6);\draw[snake=snake,segment amplitude=1] (7,6.2) -- (9,6.2);
\draw[snake=snake,segment amplitude=1] (7,6.4) -- (9,6.4);\draw[snake=snake,segment amplitude=1] (7,6.6) -- (9,6.6);
\draw[snake=snake,segment amplitude=1] (7,6.8) -- (9,6.8);
\end{scope}

\filldraw[white, draw=black] (9,4) rectangle (17.5,1);
\draw[snake=snake,segment amplitude=1] (9,1.2) -- (17.5,1.2);\draw[snake=snake,segment amplitude=1] (9,2.2) -- (17.5,2.2);
\draw[snake=snake,segment amplitude=1] (9,1.4) -- (17.5,1.4);\draw[snake=snake,segment amplitude=1] (9,2.4) -- (17.5,2.4);
\draw[snake=snake,segment amplitude=1] (9,1.6) -- (17.5,1.6);\draw[snake=snake,segment amplitude=1] (9,2.6) -- (17.5,2.6);
\draw[snake=snake,segment amplitude=1] (9,1.8) -- (17.5,1.8);\draw[snake=snake,segment amplitude=1] (9,2.8) -- (17.5,2.8);
\draw[snake=snake,segment amplitude=1] (9,2) -- (17.5,2);
\draw[snake=snake,segment amplitude=1] (9,3) -- (17.5,3);\draw[snake=snake,segment amplitude=1] (9,3.2) -- (17.5,3.2);
\draw[snake=snake,segment amplitude=1] (9,3.4) -- (17.5,3.4);\draw[snake=snake,segment amplitude=1] (9,3.6) -- (17.5,3.6);
\draw[snake=snake,segment amplitude=1] (9,3.8) -- (17.5,3.8);

\draw (2,1) -- (18,1);
\draw (2.5,1) node[anchor=north,yshift=-.5] {$0$} -- (2.5,7) -- (5,7) -- (5,5) -- (7,5)  -- (7,2) -- (9,2) -- (9,1)
node[anchor=north,yshift=-.5] {$m$};
\draw (9,4) -- (17.5,4) -- (17.5,1) node[anchor=north,yshift=-.5] {$1$};
\node at (10,10) {$f$};
\node at (6,3.5) {$f_1$};
\node[anchor=north,yshift=-.5] at (7,1) {$t$};
\draw[<-] (4,9) -- (7,9) node[anchor=west,xshift=1] {$f_r$};
\draw[dashed] (7,1) -- (7,5);
\end{scope}

\begin{scope}[shift={(0,-12)}]
  \filldraw[fill=gray,draw=gray] (2.5,1) rectangle  (5,7);
 \filldraw[fill=gray,draw=gray]  (5,1) rectangle (7,6);
\filldraw[dashed,fill=none,pattern=north west lines] (5,5) rectangle (7,6);

\filldraw[white, draw=black] (7,4) rectangle (17.5,1);
\draw[snake=snake,segment amplitude=1] (9,1.2) -- (17.5,1.2);\draw[snake=snake,segment amplitude=1] (9,2.2) -- (17.5,2.2);
\draw[snake=snake,segment amplitude=1] (9,1.4) -- (17.5,1.4);\draw[snake=snake,segment amplitude=1] (9,2.4) -- (17.5,2.4);
\draw[snake=snake,segment amplitude=1] (9,1.6) -- (17.5,1.6);\draw[snake=snake,segment amplitude=1] (9,2.6) -- (17.5,2.6);
\draw[snake=snake,segment amplitude=1] (9,1.8) -- (17.5,1.8);\draw[snake=snake,segment amplitude=1] (9,2.8) -- (17.5,2.8);
\draw[snake=snake,segment amplitude=1] (9,2) -- (17.5,2);
\draw[snake=snake,segment amplitude=1] (9,3) -- (17.5,3);\draw[snake=snake,segment amplitude=1] (9,3.2) -- (17.5,3.2);
\draw[snake=snake,segment amplitude=1] (9,3.4) -- (17.5,3.4);\draw[snake=snake,segment amplitude=1] (9,3.6) -- (17.5,3.6);
\draw[snake=snake,segment amplitude=1] (9,3.8) -- (17.5,3.8);

\filldraw[fill=white,draw=black] (2.5,7) rectangle (5,10);
\draw[snake=snake,segment amplitude=1] (2.5,7.2) -- (5,7.2);\draw[snake=snake,segment amplitude=1] (2.5,8.2) -- (5,8.2);
\draw[snake=snake,segment amplitude=1] (2.5,7.4) -- (5,7.4);\draw[snake=snake,segment amplitude=1] (2.5,8.4) -- (5,8.4);
\draw[snake=snake,segment amplitude=1] (2.5,7.6) -- (5,7.6);\draw[snake=snake,segment amplitude=1] (2.5,8.6) -- (5,8.6);
\draw[snake=snake,segment amplitude=1] (2.5,7.8) -- (5,7.8);\draw[snake=snake,segment amplitude=1] (2.5,8.8) -- (5,8.8);
\draw[snake=snake,segment amplitude=1] (2.5,8) -- (5,8);\draw[snake=snake,segment amplitude=1] (2.5,9.2) -- (5,9.2);
\draw[snake=snake,segment amplitude=1] (2.5,9.4) -- (5,9.4);\draw[snake=snake,segment amplitude=1] (2.5,9.6) -- (5,9.6);
\draw[snake=snake,segment amplitude=1] (2.5,9.8) -- (5,9.8);\draw[snake=snake,segment amplitude=1] (2.5,9) -- (5,9);

\begin{scope}[shift={(0,1)}]
 \filldraw[white, draw=black] (5,5) rectangle (7,8);
\draw[snake=snake,segment amplitude=1] (5,5.2) -- (7,5.2);\draw[snake=snake,segment amplitude=1] (5,6.2) -- (7,6.2);
\draw[snake=snake,segment amplitude=1] (5,5.4) -- (7,5.4);\draw[snake=snake,segment amplitude=1] (5,6.4) -- (7,6.4);
\draw[snake=snake,segment amplitude=1] (5,5.6) -- (7,5.6);\draw[snake=snake,segment amplitude=1] (5,6.6) -- (7,6.6);
\draw[snake=snake,segment amplitude=1] (5,5.8) -- (7,5.8);\draw[snake=snake,segment amplitude=1] (5,6.8) -- (7,6.8);
\draw[snake=snake,segment amplitude=1] (5,6) -- (7,6);
\draw[snake=snake,segment amplitude=1] (5,7) -- (7,7);\draw[snake=snake,segment amplitude=1] (5,7.2) -- (7,7.2);
\draw[snake=snake,segment amplitude=1] (5,7.4) -- (7,7.4);\draw[snake=snake,segment amplitude=1] (5,7.6) -- (7,7.6);
\draw[snake=snake,segment amplitude=1] (5,7.8) -- (7,7.8);
\end{scope}

\begin{scope}[shift={(0,-3)}]
\draw[snake=snake,segment amplitude=1] (7,4.2) -- (9,4.2);\draw[snake=snake,segment amplitude=1] (7,5.2) -- (9,5.2);
\draw[snake=snake,segment amplitude=1] (7,4.4) -- (9,4.4);\draw[snake=snake,segment amplitude=1] (7,5.4) -- (9,5.4);
\draw[snake=snake,segment amplitude=1] (7,4.6) -- (9,4.6);\draw[snake=snake,segment amplitude=1] (7,5.6) -- (9,5.6);
\draw[snake=snake,segment amplitude=1] (7,4.8) -- (9,4.8);\draw[snake=snake,segment amplitude=1] (7,5.8) -- (9,5.8);
\draw[snake=snake,segment amplitude=1] (7,5) -- (9,5);
\draw[snake=snake,segment amplitude=1] (7,6) -- (9,6);\draw[snake=snake,segment amplitude=1] (7,6.2) -- (9,6.2);
\draw[snake=snake,segment amplitude=1] (7,6.4) -- (9,6.4);\draw[snake=snake,segment amplitude=1] (7,6.6) -- (9,6.6);
\draw[snake=snake,segment amplitude=1] (7,6.8) -- (9,6.8);
\end{scope}

\draw (2,1) -- (18,1);
\draw (2.5,1) node[anchor=north,yshift=-.5] {$0$} -- (2.5,7) -- (5,7) -- (5,6) -- (7,6)  -- (7,1) -- (9,1)
node[anchor=north,yshift=-.5] {$m$};
\draw (9,4) -- (17.5,4) -- (17.5,1) node[anchor=north,yshift=-.5] {$1$};
\node at (10,10) {$f'$};
\node at (6,3.5) {$f_1$};
\node[anchor=north,yshift=-.5] at (7,1) {$t$};
\draw[dashed] (7,1) -- (7,5);
\end{scope}

\begin{scope}[shift={(17,9)}]
 \node[anchor=north,yshift=-.5] at (7,-8) {$t$};
 \filldraw[fill=gray,draw=gray] (2.5,-8) rectangle  (5,-2);
 \filldraw[fill=gray,draw=gray]  (5,-8) rectangle (7,-4);
 \filldraw[fill=gray,draw=gray]  (7,-7) rectangle (9,-8);
 
 \draw (7,-3) rectangle (9,-7);
 \draw (7,-3) rectangle (5,-4);
 \draw (9,-3) rectangle (17.5,-8);

\draw (2,-8) -- (18,-8);
\draw (2.5,-8) node[anchor=north,yshift=-.5] {$0$} -- (2.5,-2) -- (5,-2) -- (5,-4) -- (7,-4) -- (7,-7) -- (9,-7) -- (9,-8)
 node[anchor=north,yshift=-.5] {$m$};

\draw[snake=snake,segment amplitude=1] (9,-6.2) -- (17.5,-6.2);
\draw[snake=snake,segment amplitude=1] (9,-6.4) -- (17.5,-6.4);
\draw[snake=snake,segment amplitude=1] (9,-6.6) -- (17.5,-6.6);
\draw[snake=snake,segment amplitude=1] (9,-6.8) -- (17.5,-6.8);
\draw[snake=snake,segment amplitude=1] (9,-7) -- (17.5,-7);
\draw[snake=snake,segment amplitude=1] (9,-7.2) -- (17.5,-7.2);
\draw[snake=snake,segment amplitude=1] (9,-7.4) -- (17.5,-7.4);
\draw[snake=snake,segment amplitude=1] (9,-7.6) -- (17.5,-7.6);
\draw[snake=snake,segment amplitude=1] (9,-7.8) -- (17.5,-7.8);

\draw[snake=snake,segment amplitude=1] (9,-3.2) -- (17.5,-3.2);\draw[snake=snake,segment amplitude=1] (9,-3.4) -- (17.5,-3.4);
\draw[snake=snake,segment amplitude=1] (9,-3.8) -- (17.5,-3.8);\draw[snake=snake,segment amplitude=1] (9,-3.6) -- (17.5,-3.6);
\draw[snake=snake,segment amplitude=1] (9,-4) -- (17.5,-4);
\draw[snake=snake,segment amplitude=1] (9,-4.2) -- (17.5,-4.2);\draw[snake=snake,segment amplitude=1] (9,-4.4) -- (17.5,-4.4);
\draw[snake=snake,segment amplitude=1] (9,-4.8) -- (17.5,-4.8);\draw[snake=snake,segment amplitude=1] (9,-4.6) -- (17.5,-4.6);
\draw[snake=snake,segment amplitude=1] (9,-5) -- (17.5,-5);
\draw[snake=snake,segment amplitude=1] (9,-5.2) -- (17.5,-5.2);\draw[snake=snake,segment amplitude=1] (9,-5.4) -- (17.5,-5.4);
\draw[snake=snake,segment amplitude=1] (9,-5.8) -- (17.5,-5.8);\draw[snake=snake,segment amplitude=1] (9,-5.6) -- (17.5,-5.6);
\draw[snake=snake,segment amplitude=1] (9,-6) -- (17.5,-6);

\draw[snake=snake,segment amplitude=1] (7,-4.8) -- (9,-4.8);\draw[snake=snake,segment amplitude=1] (7,-4.6) -- (9,-4.6);
\draw[snake=snake,segment amplitude=1] (7,-4.4) -- (9,-4.4);\draw[snake=snake,segment amplitude=1] (7,-4.2) -- (9,-4.2);
\draw[snake=snake,segment amplitude=1] (7,-4) -- (9,-4);
\draw[snake=snake,segment amplitude=1] (7,-3.8) -- (9,-3.8);\draw[snake=snake,segment amplitude=1] (7,-3.6) -- (9,-3.6);
\draw[snake=snake,segment amplitude=1] (7,-3.4) -- (9,-3.4);\draw[snake=snake,segment amplitude=1] (7,-3.2) -- (9,-3.2);
\draw[snake=snake,segment amplitude=1] (7,-5.8) -- (9,-5.8);\draw[snake=snake,segment amplitude=1] (7,-5.6) -- (9,-5.6);
\draw[snake=snake,segment amplitude=1] (7,-5.4) -- (9,-5.4);\draw[snake=snake,segment amplitude=1] (7,-5.2) -- (9,-5.2);
\draw[snake=snake,segment amplitude=1] (7,-6) -- (9,-6);\draw[snake=snake,segment amplitude=1] (7,-5) -- (9,-5);
\draw[snake=snake,segment amplitude=1] (7,-6.8) -- (9,-6.8);\draw[snake=snake,segment amplitude=1] (7,-6.6) -- (9,-6.6);
\draw[snake=snake,segment amplitude=1] (7,-6.4) -- (9,-6.4);\draw[snake=snake,segment amplitude=1] (7,-6.2) -- (9,-6.2);

\draw[snake=snake,segment amplitude=1] (5,-3.8) -- (7,-3.8);\draw[snake=snake,segment amplitude=1] (5,-3.6) -- (7,-3.6);
\draw[snake=snake,segment amplitude=1] (5,-3.4) -- (7,-3.4);\draw[snake=snake,segment amplitude=1] (5,-3.2) -- (7,-3.2);

\draw[dashed] (7,-8) -- (7,-5);

\filldraw[pattern=north west lines,dashed]  (7,-8) rectangle (9,-7);
\filldraw[pattern=north west lines,dashed]  (5,-4) rectangle (7,-3);

\node at (10,1) {$g$};
\node at (5,-5.5) {$g_1 = f_1$};
\node[anchor=north,yshift=-.5] at (17.5,-8) {$1$};
\end{scope}

\begin{scope}[shift={(17,-3)}]
 \node[anchor=north,yshift=-.5] at (7,-8) {$t$};
 \filldraw[fill=gray,draw=gray] (2.5,-8) rectangle  (5,-2);
 \filldraw[fill=gray,draw=gray]  (5,-8) rectangle (7,-3);
 
 \draw (7,-3) rectangle (17.5,-8);

\draw (2,-8) -- (18,-8);
\draw (2.5,-8) node[anchor=north,yshift=-.5] {$0$} -- (2.5,-2) -- (5,-2) -- (5,-3) -- (7,-3) -- (7,-8) -- (9,-8)
 node[anchor=north,yshift=-.5] {$m$};

\draw[snake=snake,segment amplitude=1] (9,-6.2) -- (17.5,-6.2);
\draw[snake=snake,segment amplitude=1] (9,-6.4) -- (17.5,-6.4);
\draw[snake=snake,segment amplitude=1] (9,-6.6) -- (17.5,-6.6);
\draw[snake=snake,segment amplitude=1] (9,-6.8) -- (17.5,-6.8);
\draw[snake=snake,segment amplitude=1] (9,-7) -- (17.5,-7);
\draw[snake=snake,segment amplitude=1] (9,-7.2) -- (17.5,-7.2);
\draw[snake=snake,segment amplitude=1] (9,-7.4) -- (17.5,-7.4);
\draw[snake=snake,segment amplitude=1] (9,-7.6) -- (17.5,-7.6);
\draw[snake=snake,segment amplitude=1] (9,-7.8) -- (17.5,-7.8);

\draw[snake=snake,segment amplitude=1] (9,-3.2) -- (17.5,-3.2);\draw[snake=snake,segment amplitude=1] (9,-3.4) -- (17.5,-3.4);
\draw[snake=snake,segment amplitude=1] (9,-3.8) -- (17.5,-3.8);\draw[snake=snake,segment amplitude=1] (9,-3.6) -- (17.5,-3.6);
\draw[snake=snake,segment amplitude=1] (9,-4) -- (17.5,-4);
\draw[snake=snake,segment amplitude=1] (9,-4.2) -- (17.5,-4.2);\draw[snake=snake,segment amplitude=1] (9,-4.4) -- (17.5,-4.4);
\draw[snake=snake,segment amplitude=1] (9,-4.8) -- (17.5,-4.8);\draw[snake=snake,segment amplitude=1] (9,-4.6) -- (17.5,-4.6);
\draw[snake=snake,segment amplitude=1] (9,-5) -- (17.5,-5);
\draw[snake=snake,segment amplitude=1] (9,-5.2) -- (17.5,-5.2);\draw[snake=snake,segment amplitude=1] (9,-5.4) -- (17.5,-5.4);
\draw[snake=snake,segment amplitude=1] (9,-5.8) -- (17.5,-5.8);\draw[snake=snake,segment amplitude=1] (9,-5.6) -- (17.5,-5.6);
\draw[snake=snake,segment amplitude=1] (9,-6) -- (17.5,-6);

\draw[snake=snake,segment amplitude=1] (7,-4.8) -- (9,-4.8);\draw[snake=snake,segment amplitude=1] (7,-4.6) -- (9,-4.6);
\draw[snake=snake,segment amplitude=1] (7,-4.4) -- (9,-4.4);\draw[snake=snake,segment amplitude=1] (7,-4.2) -- (9,-4.2);
\draw[snake=snake,segment amplitude=1] (7,-4) -- (9,-4);
\draw[snake=snake,segment amplitude=1] (7,-3.8) -- (9,-3.8);\draw[snake=snake,segment amplitude=1] (7,-3.6) -- (9,-3.6);
\draw[snake=snake,segment amplitude=1] (7,-3.4) -- (9,-3.4);\draw[snake=snake,segment amplitude=1] (7,-3.2) -- (9,-3.2);
\draw[snake=snake,segment amplitude=1] (7,-5.8) -- (9,-5.8);\draw[snake=snake,segment amplitude=1] (7,-5.6) -- (9,-5.6);
\draw[snake=snake,segment amplitude=1] (7,-5.4) -- (9,-5.4);\draw[snake=snake,segment amplitude=1] (7,-5.2) -- (9,-5.2);
\draw[snake=snake,segment amplitude=1] (7,-6) -- (9,-6);\draw[snake=snake,segment amplitude=1] (7,-5) -- (9,-5);
\draw[snake=snake,segment amplitude=1] (7,-6.8) -- (9,-6.8);\draw[snake=snake,segment amplitude=1] (7,-6.6) -- (9,-6.6);
\draw[snake=snake,segment amplitude=1] (7,-6.4) -- (9,-6.4);\draw[snake=snake,segment amplitude=1] (7,-6.2) -- (9,-6.2);
\draw[snake=snake,segment amplitude=1] (7,-7.8) -- (9,-7.8);\draw[snake=snake,segment amplitude=1] (7,-7.6) -- (9,-7.6);
\draw[snake=snake,segment amplitude=1] (7,-7.4) -- (9,-7.4);\draw[snake=snake,segment amplitude=1] (7,-7.2) -- (9,-7.2);
\draw[snake=snake,segment amplitude=1] (7,-7) -- (9,-7);

\filldraw[pattern=north west lines,dashed]  (7,-8) rectangle (9,-7);
\filldraw[pattern=north west lines,dashed]  (5,-4) rectangle (7,-3);


\node at (10,1) {$g'$};
\node at (5,-5.5) {$g_1 = f_1$};
\node[anchor=north,yshift=-.5] at (17.5,-8) {$1$};
\end{scope}

       \end{tikzpicture}
    \end{minipage}%
          \caption{An example of two CFPs: $(f,g)$ is produced by Lemma\nobreakspace \ref {lem:introduce_m}, whereas $(f',g')$ is produced by Lemma\nobreakspace \ref {lem:introduce_t}. In this picture, the height of the plot corresponds to $\Size{f(x)}$ for each $x$, while the shaded and wavy parts correspond to the contributions of $f_1$ and $f_r$ to $\size_f$; similarly for $g$. The wavy parts are liquid.
          In~Lemma\nobreakspace \ref {lem:introduce_t} we want to make $f_1 = g_1$ equal to $\size_g$ on an interval $[0,t)$, so we increase sizes of solid
elements in $g$ on that interval, while decreasing those on the interval $[t,m)$. The striped regions in the pictures of $g$ correspond to
these changes. (We repeat the same changes in $f$, and we also move liquid elements in $g$ to keep $\size_g$ unchanged.)
          The threshold $t \in [0,m]$ is chosen so that $g_1$ becomes equal to $\size_g$ on $[0,t)$ while the solid elements on $[t,m)$ are eradicated (i.e., so that the areas of the striped regions in $g$ are equal).}
    
    \label{fig:idealdistro}

\end{figure*}

%% file: hardness.tex

\section{Integrality Gap Lower Bound}\label{app:conflp-integrality-gap}
\input{igap}
First of all, observe that Theorem \ref{thm:main}, apart from establishing the existence of a 1.21-approximation algorithm for
$R | | \sum_j p_j C_j$, also implies an upper bound on the integrality gap of its Configuration-LP. Hence, we accompany our
main result with a lower bound on the integrality gap of the Configuration-LP for $R | | \sum_j p_j C_j$:
\ifsoda
\else
	\begin{artificialtheorem}[Theorem\nobreakspace \ref *{thm:gap}.]
		\integralitygapconflp
	\end{artificialtheorem}
\fi

\begin{proof}
Consider the following instance on $4$ machines $M_1, M_2, M_3, M_4$: for every pair $\{M_i,M_j\}$ of machines there is one job $J_{ij}$
which can be processed only on these two machines. Jobs $J_{12}$ and $J_{34}$ are \emph{large}: they have weight and size $3$, while the
other four jobs are \emph{small} and have weight and size $1$. (See Figure\nobreakspace \ref {fig:igap} for an illustration.)

First we show that any integral schedule has cost at least $26$. Without loss of generality, the large job $J_{12}$ is assigned to machine
$M_1$ and the other large job $J_{34}$ is assigned to machine $M_3$. The small job $J_{13}$ must also be assigned to one of them, say to
$M_1$. This costs $9+9+4=22$. The remaining three small jobs $J_{12}$, $J_{14}$ and $J_{24}$ cannot all be assigned to distinct machines
with zero makespan (since only $M_2$ and $M_4$ are such), so they will incur at least $1+1+2=4$ units of cost.

On the other hand, the Configuration-LP has a solution of cost $24$. Namely, it can assign to each machine $M_i$ two configurations, each
with fractional value $\frac 12$: the singleton large job that can be processed on that machine, or the two small jobs that can be processed
on that machine. Then each job is processed with fractional value $\frac 12$ by each of the two machines that can process it. The cost is $4
\cdot (\frac 12 \cdot 9 + \frac 12 \cdot (1 + 2)) = 24$. Thus the integrality gap is at least $\frac{26}{24} = \frac{13}{12} > 1.08$.
\end{proof}


%% file: igap.tex
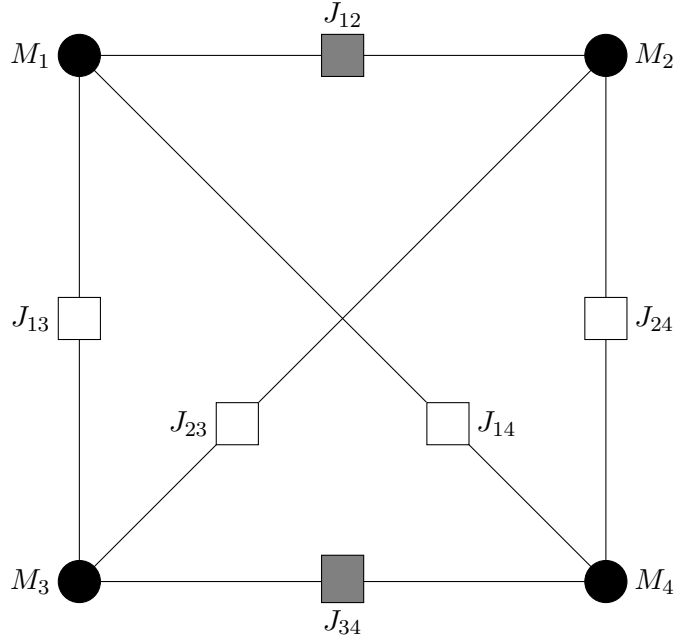
\begin{figure*}[t!]
\begin{minipage}[t]{\linewidth}
    \begin{tikzpicture}[scale=0.7]
     \node at (-7,0) {};
\draw (0,0) -- (10,0);
\draw (0,10) -- (10,10);
\draw (0,0) -- (0,10);
\draw (10,0) -- (10,10);
\draw (0,0) -- (10,10);
\draw (10,0) -- (0,10);
     
\filldraw  (0,0) circle (.4);
\filldraw  (0,10) circle (.4);
\filldraw  (10,0) circle (.4);
\filldraw  (10,10) circle (.4);

\filldraw[draw=black,fill=gray] (4.6,-.4) rectangle (5.4,.5);
\filldraw[draw=black,fill=gray] (4.6,9.6) rectangle (5.4,10.4);
\filldraw[draw=black,fill=white] (-.4,4.6) rectangle (.4,5.4);
\filldraw[draw=black,fill=white] (9.6,4.6) rectangle (10.4,5.4);
\filldraw[draw=black,fill=white] (2.6,2.6) rectangle (3.4,3.4);
\filldraw[draw=black,fill=white] (6.6,2.6) rectangle (7.4,3.4);

\node[anchor=east,xshift=-7] at (0,0) {$M_3$};
\node[anchor=east,xshift=-7] at (0,10) {$M_1$};
\node[anchor=west,xshift=7] at (10,0) {$M_4$};
\node[anchor=west,xshift=7] at (10,10) {$M_2$};
\node[anchor=south,yshift=7] at (5,10) {${J_{12}}$};
\node[anchor=north,yshift=-7] at (5,0) {${J_{34}}$};
\node[anchor=east,xshift=-7] at (0,5) {${J_{13}}$};
\node[anchor=west,xshift=7] at (10,5) {${J_{24}}$};
\node[anchor=east,xshift=-7] at (3,3) {${J_{23}}$};
\node[anchor=west,xshift=7] at (7,3) {${J_{14}}$};
     
     \end{tikzpicture}
    \end{minipage}%
          \caption{The $\frac{13}{12}$-integrality gap instance. In this picture, black circles correspond to machines, gray boxes
correspond to jobs of size 3, and white boxes correspond to jobs of size 1. An edge between a circle and a box means that the corresponding
job can be assigned to the corresponding machine.}
    
    \label{fig:igap}

\end{figure*}